%% file: robust_dim_testing.tex
\documentclass[11pt]{article}

\usepackage{fullpage}
\usepackage{amsmath,amsfonts,amsthm,xspace,hyperref,graphicx,relsize,bm}
\usepackage[margin=1in]{geometry}
\usepackage{float}
\usepackage{endnotes}
\usepackage{enumitem}
\usepackage{color}
\allowdisplaybreaks

\newtheorem{theorem}{Theorem}
\newtheorem{proposition}[theorem]{Proposition}
\newtheorem{lemma}[theorem]{Lemma}
\newtheorem{claim}[theorem]{Claim}

\newtheorem{definition}[theorem]{Definition}
\theoremstyle{definition}

\usepackage[normalem]{ulem}
\usepackage{todonotes}

\newtheorem{innercustomthm}{Theorem}
\newenvironment{customcorollary}[1]
  {\renewcommand\theinnercustomthm{#1}\innercustomthm}
  {\endinnercustomthm}

\input{qmacros}

\begin{document}

\title{Noise-tolerant testing of high entanglement of formation}
\author{Rotem Arnon-Friedman\thanks{RAF is supported by the Swiss National Science Foundation (grant No. 200020-135048) via the National Centre of Competence in Research ``Quantum Science and Technology'' and by the US Air Force Office of Scientific Research (grant No. FA9550-16-1-0245).}\\ ETH-Z\"{u}rich \and Henry Yuen\thanks{HY is supported by ARO Grant W911NF-12-1-0541 and NSF Grant CCF-1410022. } \\ UC Berkeley}

\date{\empty}
\maketitle

\newcommand{\mj}{{-j}}
\newcommand{\qval}{\mathrm{qval}}
\newcommand{\cval}{\mathrm{cval}}

\newcommand{\DivMid}{\, \Big \| \,}

\newcommand{\Qvec}{\mathbf{Q}}
\newcommand{\Rvec}{\mathbf{R}}
\newcommand{\Dvec}{\mathbf{D}}
\newcommand{\Dbar}{\overline{D}}
\newcommand{\Cbar}{\overline{C}}
\newcommand{\Sbar}{\overline{S}}

\newcommand{\fulla}{\avec}
\newcommand{\fullb}{\bvec}
\newcommand{\aC}{\avec_C}
\newcommand{\bC}{\bvec_C}
\newcommand{\X}{\mathsf{X}}
\newcommand{\Y}{\mathsf{Y}}
\newcommand{\Z}{\mathsf{Z}}
\newcommand{\A}{\mathsf{A}}
\newcommand{\B}{\mathsf{B}}
\newcommand{\W}{\mathsf{W}}
\newcommand{\U}{\mathsf{U}}
\newcommand{\Hilb}{\mathcal{H}}
\newcommand{\mi}{{-i}}

\newcommand{\wt}[1]{\widetilde{#1}}
\newcommand{\what}[1]{\widehat{#1}}
\newcommand{\eval}{\mathrm{val}^*}
\newcommand{\val}{\mathrm{val}}
\newcommand{\ac}{\mathrm{Z}}

\newcommand{\xvec}{\mathbf{x}}
\newcommand{\yvec}{\mathbf{y}}
\newcommand{\avec}{\mathbf{a}}
\newcommand{\bvec}{\mathbf{b}}

\newcommand{\Xvec}{\mathbf{X}}
\newcommand{\Yvec}{\mathbf{Y}}
\newcommand{\Avec}{\mathbf{A}}
\newcommand{\Bvec}{\mathbf{B}}

\newcommand{\comp}[1]{\overline{#1}}

\newcommand{\ScupT}{{S \cup T}}
\newcommand{\theevent}{W^{\geq 1 - \tau}_S}

\begin{abstract}
In this work we construct tests that allow a classical user to certify high dimensional entanglement in uncharacterized and possibly noisy quantum devices.
We present a family of non-local games $\{G_n\}$ that for all $n$ certify states with entanglement of formation $\Omega(n)$.
These tests can be derived from \emph{any} bipartite non-local game with a classical-quantum gap.
Furthermore, our tests are noise-tolerant in the sense that fault tolerant technologies are not needed to play the games; entanglement distributed over noisy channels can pass with high probability, making our tests relevant for realistic experimental settings. 
This is in contrast to, e.g., results on self-testing of high dimensional entanglement, which are only relevant when the noise rate goes to zero with the system's size $n$.  
As a corollary of our result, we supply a lower-bound on the entanglement cost of any state achieving a quantum advantage in a bipartite non-local game.
Our proof techniques heavily rely on ideas from the work on classical and quantum parallel repetition theorems.

\end{abstract}

\section{Introduction}

Non-local games offer a powerful method to experimentally study the properties and behavior of uncharacterized quantum systems. In a non-local game, an experimenter can play a game with two non-communicating players (representing spatially separated quantum systems) via classical interaction only. Based on the outcome of the game, the experimenter draws conclusions about, e.g., whether the players used an entangled quantum state to win the game. This idea dates back to John Bell's seminal paper~\cite{bell1964}, in which he presents a game to test the non-classicality of nature. Today, such games are not only relevant for our understanding of the foundations of quantum physics but are at the heart of device-independent quantum information processing, where a classical user can certify that an unknown quantum device is performing a desired computational or cryptographic task (such as, e.g., device-independent quantum key distribution~\cite{barrett2005no,pironio2009device,vazirani2014fully,miller2014robust,arnon2016simple} or delegated quantum computation~\cite{reichardt2013classical,hajduvsek2015device,gheorghiu2015robustness,natarajan2017quantum,coladangelo2017verifier}).

In this work we ask the following question: 
\begin{center}
	\emph{Is it possible to classically test for high dimensional entanglement, even in the presence of noise?}
\end{center}
Whereas Bell's original test is a classical method to certify the \emph{presence} of entanglement, we are instead interested in non-local games that would allow us to quantify the \emph{amount}. In particular, we are interested in certifying the amount of entanglement of \emph{noisy quantum systems}.

Designing noise-tolerant tests for high dimensional entanglement is an important and timely challenge for both computer science and physics. First, our understanding of complexity theory indicates that unless ${\sf BQP} \subseteq {\sf BPP}$ (i.e., quantum computers are classically simulable), general quantum computations must involve highly entangled states. Thus if we hope to achieve super-classical speedups in quantum computers, at the very least we must be able to generate high dimensional entanglement. 

Second, we are seeing increasingly sophisticated experiments involving quantum information, from loophole-free Bell tests~\cite{hensen2015loophole,shalm2015strong,giustina2015significant} to small scale quantum computers~\cite{boixo2016characterizing,IBM}. However, full-fledged quantum fault tolerance appears to be a faraway prospect; in the near-term, our explorations of complex quantum states will be done using noisy gates and little (if any) error correction. Despite this obstacle, researchers have been enthusiastically proposing uses of noisy quantum computers, from approximate optimization to investigation of exotic physics phenomena. Interesting questions will emerge in tandem with these efforts, namely: how can one verify that a noisy quantum computer has succeeded in these proposed experiments? Finding noise-tolerant tests to certify high dimensional entanglement is a prerequisite step towards verifying other complex quantum behavior in this noisy regime.

\paragraph{What do we mean by certifying entanglement?} There are a variety of ways to formulate this task; our work is most directly motivated by recent work on \emph{self-tests}, which are games that certify the presence of entanglement \emph{of a specific form}. The works of~\cite{mckague2016self,chao2016test,coladangelo2017parallel,coudron2016parallel,natarajan2017quantum} construct families of games $\{G_n\}$ where any optimal quantum strategy for $G_n$ must use a large amount of entanglement, e.g., a tensor product of $n$ EPR pairs. These self-testing results are also \emph{robust}, in that near-optimal strategies must use states that are near a specific highly entangled state. However, these tests will also reject a natural class of highly entangled states such as $\sigma^{\otimes n}$ where $\sigma$ has fidelity $1 - \nu$ with a single EPR pair. Here, think of $\nu$ as a small (but fixed) noise parameter that represents the level of imperfection of a state preparation process.

Thus, even though $\ketbra{EPR}{EPR}^{\otimes n}$ can be used to pass the tests of~\cite{mckague2016self,chao2016test,coladangelo2017parallel,coudron2016parallel,natarajan2017quantum} with high probability, the ``similar-looking'' state  $\sigma^{\otimes n}$ will fail with high probability. A key observation we wish to emphasize in this paper is that robustness of a self-test is \emph{not} equivalent to noise tolerance!

More formally, the robust self-tests in the above works show the following: let $\qval(G_n)$ denote the optimal quantum winning probability for the game $G_n$. Then there exists a function $f(n,\eps)$ and an \emph{ideal state} $\rho^*_n$ such that for all $\eps$, any quantum strategy that achieves a winning probability of at least $\qval(G_n) - \eps$ must use a state $\rho$ that is $f(n,\eps)$-close to $\rho^*_n$. In these works, $\rho^*_n$ is a state whose entanglement grows with $n$ (like a maximally entangled state on $n$ qubits). ``Closeness'' can be defined in terms of the fidelity of the two states up to local isometries acting on each of the players' systems.

Given a game $G_n$ as above, an experiment to test the entanglement of an unknown state $\rho$ can be the following: play the game $G_n$ using $\rho$, and check whether the game is won.\footnote{In an experiment one actually needs to prepare many identical and independent copies of $\rho$ and play the game $G_n$ many times. Then the average winning probability can be calculated, and high amount of entanglement is certified (with high probability) if the average winning probability is at least $\qval(G_n) - \eps$.} 
In order to obtain a non-trivial guarantee about $\rho$, we require that $f(n,\eps) < 1$; one can think of this function as specifying the amount of experimental imperfection/noise that can be tolerated by the \emph{test} itself. In the works of~\cite{mckague2016self,chao2016test,coladangelo2017parallel,coudron2016parallel}, the function $f(n,\eps)$ scales as $a \cdot n^b \cdot \eps^c$ for constants $a,b,c$. Thus we get no guarantees about $\rho$ unless $\eps$ scales as $1/\poly(n)$. In other words, as we increase the amount of entanglement we want to certify, the test becomes less tolerant of noise! 

The strongest self-testing result (in this context) is presented in the work of Natarajan and Vidick~\cite{natarajan2017quantum}. There, a self-test for $n$ EPR pairs is given where the associated function is $f(n,\eps) = O(\sqrt{\eps})$. 
While the closeness parameter is independent of the parameter $n$, such $f(n,\eps)$ still requires that, in order to pass the test with high probability, the players share a state $\rho$ that is \emph{globally} $O(\sqrt{\eps})$-close to $\ketbra{EPR}{EPR}^{\otimes n}$. Using a state like $\sigma^{\otimes n}$  where $\sigma$ has $1 - \nu$ fidelity with a single EPR pair would fail their test with high probability, because $\sigma^{\otimes n}$  has \emph{exponentially small} fidelity $(1 - \nu)^n \approx e^{-n/\nu}$ with $\ketbra{EPR}{EPR}^{\otimes n}$.

In this paper we seek an entanglement test that is both \emph{sound} --- meaning that any strategy that passes the test with good probability must have high entanglement --- and also \emph{noise tolerant}, meaning that they do not reject noisy implementations of an ideal strategy. The self-tests above are sound, but they are not noise tolerant. Part of the difficulty stems from the fact that it is not even clear how one should \emph{formulate} the soundness guarantee of a desired noise-tolerant self-testing result.

\paragraph{Noise model.} 

As discussed above, we wish to define a testing procedure that can also certify entanglement in noisy entangled states. While our work can be used to certify different types of noisy states, we briefly discuss a specific noise model here for the sake of concreteness. 
The noise model that we have in mind produces a state of the form $\sigma^{\otimes n}$ where each $\sigma$ has fidelity $1-\nu$ with some optimal state defined via the considered non-local game.  Such a state can be produced, e.g., by sending many copies of the optimal state via noisy channels. 

We emphasize that by saying that this is the noise model that we consider we merely mean that we require that our tests will be able to certify the entanglement of $\sigma^{\otimes n}$. However, we do not assume that all of the states on which the procedure is applied must have this form (i.e., the soundness part of the statement is independent of the considered noise model).

\subsection{Results and contributions}
In this work, instead of trying to certify the presence of a specific state like in self-testing statements, we address the question of certifying an entanglement measure. This allows us to sidestep the difficulty of formulating a noise-tolerant self-testing result. 

We present a family of simple non-local games $\{G_n\}$ where each game $G_n$ certifies that the shared state of the players has $\Omega(n)$ bits of \emph{entanglement of formation}. The entanglement of formation, denoted by $E_F(\rho)$, is a   well-studied entanglement measure for bipartite mixed states that, in the case of pure states, is equal to the entanglement entropy. As the name suggests, the entanglement of formation captures, roughly speaking, the amount of entanglement needed in order to produce a given state $\rho$. It is also closely related to another important, perhaps more well known, entanglement measure which will be of use below -- the \emph{entanglement cost} $E_C(\rho)$. The entanglement cost of a mixed state roughly describes how many EPR pairs are needed to create $\rho$ via local operations and classical communication~\cite{bennett1996mixed}. 
We provide a more thorough discussion of the entanglement measures relevant for our work in Section~\ref{sec:ent_of_formation}.

The family of non-local games that we consider are the so called \emph{threshold games}.
Before stating our main result, we define these games. Let $G$ be a two-player non-local game with classical value\footnote{The classical value of a game is the maximum winning probability when the players employ classical strategies, i.e., do not use entanglement. Similarly, the quantum value of a game is the optimal winning probability when using quantum strategies. See Section~\ref{sec:games_related_def} for the formal definition.} $\cval(G)$ and quantum value $\qval(G)$. Given an integer $n \geq 1$ and a noise threshold $0 \leq \nu < \qval(G) - \cval(G)$, define the threshold game $G^n_{\qval(G) - \nu}$ to be a game where the two-players now play $n$ independent instances of $G$ in parallel, and win if they win at least $\qval(G) - \nu$ fraction of instances of $G$. 

The main theorem of this paper is as follows:
\begin{theorem}[Main theorem]
\label{thm:main}
	Let $G$ be a two-player game with a classical-quantum gap: i.e., $\Delta := \qval(G) - \cval(G) > 0$. Let $0 \leq \nu < \Delta$ be a noise parameter. 

	\medskip 
		
	\noindent \textbf{Completeness (Noise tolerance).} Let $n \geq 1$ be an integer. Consider a quantum strategy for $G$ that succeeds with probability $\qval(G) - \eta$ for $0\leq\eta<\nu$. Playing this strategy $n$ times independently in parallel in the threshold game $G^n_{\qval(G) - \nu}$ succeeds with probability at least $1 - \exp \paren{-(\nu - \eta)^2 n/3}$.\footnote{Alternatively, a simpler (but slightly weaker) statement is that playing a strategy the succeeds with probability $\qval(G) - \nu$ in $G$ $n$ times independently in parallel succeeds in the threshold game $G^n_{\qval(G) - \nu}$ with probability $\frac{1}{2}$. This is sufficient for an experiment certifying entanglement.}
	
	\medskip 
	
	\noindent \textbf{Soundness (Entanglement certification).} There exists constants $0 < c_1,c_2 < 1$ such that for sufficiently large $n > \frac{1}{c_1}$, any strategy that wins the threshold game $G^{n}_{\qval(G) - \nu}$ with probability $\kappa \geq \exp(-c_1 n)$ must use a quantum state $\rho$ such that its entanglement of formation satisfies $E_F(\rho) \geq c_2\kappa^2 n$.

The constants $c_1,c_2$ depend only on $\Delta$, $\nu$, and the number of possible answers in $G$.
\end{theorem}

To gain a better understanding of our theorem we now give an example. 
Consider the famous CHSH game, which has classical value $\cval(CHSH) = 3/4$ and quantum value $\qval(CHSH) \approx 0.854$. Any strategy for winning a single instance of CHSH with probability $\qval(CHSH) - \eta$ for some parameter $0 \leq \eta < 0.1$ must use some entangled state $\sigma$. 
An ``honest'' strategy for playing the threshold game $CHSH^n_{.854 - 2\eta}$ would be to play each instance of $CHSH$ independently using $\sigma^{\otimes n}$ as the entangled resource state. Via a simple Chernoff-Hoeffding bound it is easy to see that this strategy will pass $CHSH^n_{.854 - 2\eta}$ with overwhelming probability. Thus this game is noise-tolerant. The entanglement of formation of $\sigma^{\otimes n}$ is indeed $\Omega(n)$.

But what about other strategies? Is there a state with entanglement of formation $o(n)$ that can be used to win  $CHSH^{n}_{.854 - 2\eta}$ sufficiently well? Theorem~\ref{thm:main} shows that this is not possible. 

We list several features of Theorem~\ref{thm:main}:
\begin{enumerate}
	\item It holds for \emph{any} two-player game $G$. In other words, any game with a classical-quantum gap can be ``lifted'' to another game that tests for large entanglement in a noise-tolerant manner.
	
	\item The players are able to pass our test with high probability by holding a tensor product of noisy few-qubit states (such as $\sigma^{\otimes n}$ where $\sigma$ has fidelity $1 - \nu$ with an EPR pair for any amount).
	The theorem gives non-trivial guarantees for any $0 \leq \nu < \qval(G)-\cval(G)$, i.e., it is robust to \emph{any} amount of noise up to the classical limit.

	\item It gives non-trivial guarantees even for strategies whose success probability is far from optimal; for any constant~$\kappa$, Theorem~\ref{thm:main} still guarantees that $E_F(\rho) \in \Omega(n)$.\footnote{However, the constants $c_1$ and $c_2$ are probably not optimal and can be improved.} 
	\end{enumerate}

Theorem~\ref{thm:main} thus shows that by playing the simple threshold game $G^{n}_{\qval(G) - \nu}$ with an uncharacterized device we can classically test for large amounts of entanglement (as measured by the entanglement of formation), even when the device is highly noisy, as current devices are. As far as we are aware, previous results~\cite{mckague2016self,chao2016test,coladangelo2017parallel,natarajan2017quantum,coladangelo2017robust} cannot be used to derive conclusions which are quantitively strong as Theorem~\ref{thm:main}, even when considering more complex games and proof techniques.\footnote{This is not to say that our work supersedes the mentioned works; these derived self-testing statements which certify the \emph{state} and not just its \emph{entanglement} as we do here.}

Our main theorem presented above can be easily used to derive another quantitive relation between the advantage in a non-local game $G$ and the \emph{entanglement cost} required to achieve this advantage. Specifically, we prove the following.

\begin{theorem}\label{cor:lwb_ent_single_game}
	Let $G$ be a two-player game with a classical-quantum gap: i.e., $\Delta := \qval(G) - \cval(G) > 0$. Let $0 \leq \nu < \Delta$ be a noise parameter. Then, for any state $\sigma$ that can be used to win~$G$ with probability at least $\qval(G) - \nu$, its entanglement cost satisfies $E_C(\sigma)\geq c_2/4$, where~$c_2$ is the constant from Theorem~\ref{thm:main}.
\end{theorem}
Put in other words: the minimum entanglement cost\footnote{For any $\sigma$, $E_F(\sigma) \geq E_C(\sigma)$. Thus, Theorem~\ref{cor:lwb_ent_single_game} could have been phrased in terms of the entanglement of formation as well.} needed to obtain a super-classical success probability in a non-local game only depends on the classical-quantum gap as well as the number of possible answers in the game.

As we explain in Section~\ref{sec:ent_of_formation}, even given the full description of a state $\sigma$, calculating $E_C(\sigma)$ is not easy and no ``single letter'' formula is known to describe it.  Theorem~\ref{cor:lwb_ent_single_game} gives a simple lower bound on $E_C(\sigma)$ in terms of $\sigma$'s advantage in any non-local game $G$.

The only lower-bound with a similar flavour which was known before is the one given in~\cite{verstraete2002entanglement}. There, a (tight) relation between $E_F(\sigma)$ and $\sigma$'s winning probability in the CHSH game was derived. Self-testing results can, of course, also be used to achieve similar bounds (by taking into account the continuity of the considered entanglement measures), but so far most of the results are non-trivial for a very limited amount of noise and only apply to specific two-player games. In contrast, Theorem~\ref{cor:lwb_ent_single_game} holds for any non-local game and amount of noise.

\subsection{Why entanglement of formation?}
\label{sec:ent_of_formation}

In this section we motivate and explain the relations between the entanglement measures certified by our tests in Theorem~\ref{thm:main}  and Theorem~\ref{cor:lwb_ent_single_game}.

Myriad entanglement measurements have been studied by researchers, each possessing various properties~\cite{plenio2005introduction,horodecki2009quantum}. 
For pure bipartite states $\ket{\psi}^{\A\B}$, the coarsest quantity describing entanglement is the \emph{entanglement rank}, which is simply the Schmidt rank of $\ket{\psi}$. However, this is not a very useful measure of entanglement as one can have a state arbitrarily close to a product state, yet have high entanglement rank.

A more natural measure of entanglement is the \emph{entanglement entropy} $E(\psi)$, which is the von Neumann entropy of the reduced density matrix of $\ket{\psi}$ on system~$\A$ or equivalently~$\B$~\cite{bennett1996concentrating,popescu1997thermodynamics}. In fact, the entanglement entropy is the \emph{unique} entanglement measure for pure bipartite states that satisfies a few natural axioms, such as monotonicity under local operations and classical communication (LOCC) and asymptotic continuity~\cite{horodecki2009quantum}. 

For mixed states the situation is more complicated --- there is no clear ``best'' entanglement measure.
The most natural and operational entanglement measures are considered to be the \emph{entanglement cost} $E_C$ and the \emph{distillable entanglement} $E_D$. In fact, for any entanglement measure $M$ satisfying some natural properties we have that $E_D \leq M \leq E_C$~\cite{horodecki2009quantum}. Thus the entanglement cost and distillable entanglement are in a sense ``extremal'' entanglement measures. For pure states, both $E_C$ and $E_D$ are equal to the entanglement entropy.

In the following we focus on $E_C$. Informally, the entanglement cost of a bipartite quantum state~$\rho_{AB}$ describes the number of maximally entangled states required to produce $\rho$ using only LOCC. As LOCC cannot increase entanglement, the pre-shared maximally entangled states describe the sole source of entanglement in such a process and hence quantify how entangled $\rho$ is in a meaningful way.\footnote{Another way of thinking about the operational meaning of entanglement cost is by considering the task of entanglement dilution. There, the goal is to start with initial noiseless entanglement and dilute it to create a target state $\rho$ using LOCC.} 

Formally, the entanglement cost is defined as the following asymptotic quantity:
\begin{equation*}
	E_C(\rho) = \inf \left\{ r : \lim_{n \rightarrow \infty} \left( \inf_{\Lambda} \| \rho^{\otimes n} - \Lambda(\Phi^+_{2^{rn}}) \|_1 \right) =0 \right\} \;,
\end{equation*}
where the infimum ranges over all LOCC maps $\Lambda$ and $\Phi^+_{2^{rn}}$ is the maximally entangled state of rank $2^{rn}$. 
That is, it is the maximal possible rate $r$ at which one can convert $\Phi^+_{2^{rn}}$ into $\rho^{\otimes n}$ with vanishing error in the limit $n \rightarrow \infty$. 

Computing $E_C(\rho)$ is considered to be a difficult task in general.
Due to this reason one usually considers a closely related entanglement measure called the \emph{entanglement of formation}. It is formally defined as follows~\cite{bennett1996mixed}:
\begin{equation*}
	E_F(\rho) = \inf\left\{ \sum_i p_i E(\Psi_i) : \rho=\sum_i p_i \ket{\Psi_i}\bra{\Psi_i} \right\} \;.
\end{equation*}
That is, $E_F(\rho)$ is the minimum average entanglement entropy
over all pure-state decompositions of $\rho$.

The entanglement of formation derives its relevance from its relation to the entanglement cost~$E_C(\rho)$ discussed above. It describes the rate in which maximally entangled states are converted to~$\rho$ using a specific type of LOCC protocols~\cite{wotters2002ent} (whereas $E_C(\rho)$ is the minimum over all LOCC protocols). Furthermore,~\cite{hayden2001asymptotic} showed that the entanglement cost is equal to the \emph{regularised} entanglement of formation:
\[
	E_C(\rho) = E^{\infty}_F(\rho)=\lim_{n\rightarrow\infty} (E_F(\rho^{\otimes n})/n).
\]
For some time it was conjectured that the entanglement of formation is additive and hence $E_C(\rho)=E_F(\rho)$. Today it is known that this is not the case and that the limit in the above equation is needed in general~\cite{brandao2010hastings}.

It is not known how to compute $E^{\infty}_F(\rho)$ for general $\rho$, in part because of the infinite limit. The ``single-letter'' quantity $E_F(\rho)$ does not appear to be much easier to compute because of the minimisation over all possible decompositions of $\rho$. To date, it can be done only for states with high symmetry~\cite{terhal2000entanglement,vollbrecht2001entanglement} or of low dimension~\cite{wootters1998entanglement,wootters2001entanglement,audenaert2001variational}. 
One can imagine that the task  of calculating or bounding $E_F(\rho)$ only becomes harder if one does not have full information about $\rho$ as in the scenario considered in the current work.

In the light of the above, one can see our work as giving a way to lower bound those complex entanglement measures for  an unknown state $\rho$ in a device-independent manner. Of course, this is not a general method that works for all states $\rho$, but rather it works for any state $\rho$ that can be used to gain an advantage in non-local games (or, in other words, violate some Bell inequality).
Specifically, Theorem~\ref{thm:main} gives a lower bound on $E_F$ for high dimensional (while perhaps noisy) states that can be used to pass the threshold game $G^n_{\qval(G) - \nu}$ for some two-player game $G$. 
Theorem~\ref{cor:lwb_ent_single_game} gives a lower bound on $E_C$ for any state achieving a quantum advantage in a two-player game $G$. 
In particular, for any given state one can choose the game $G$ such that the lower bounds on  $E_F$ and $E_C$ are maximal.

\subsection{Proof technique}

The proof idea is simple: if the entanglement of formation of the players' shared state in the threshold game $G^n_{\qval(G) - \nu}$  is $o(n)$ and the players win with non-negligible probability, then this strategy can be transformed into a strategy for the original game $G$ that uses \emph{no} entanglement, yet still wins with probability strictly greater than $\cval(G)$, which would be a contradiction. 

This is argued as follows. Consider a two-player game $G$ where the first player receives a question $x$ and produces answer $a$, and the second player receives question $y$ and responds with answer $b$. The players win if $V(x,y,a,b) = 1$ for some predicate $V$. Let $\qval(G) > \cval(G)$. 

Now suppose there is a quantum strategy that wins $G^n_{\qval(G) - \nu}$ with decent probability. A simple probabilistic argument implies that conditioned on an event $E$ of winning roughly $\qval(G) - \nu$ fraction of some subset $S \subseteq [n]$ of instances, the players will win the $j$'th instance with probability close to $\qval(G)$, for an average $j \in [n]$. Another way of phrasing this statement is: Let $(\Xvec_j,\Yvec_j)$ denote the questions to the two players in the $j$'th instance of $G$, and let $(\Avec_j,\Bvec_j)$ denote their answers. Let $\P_{\Xvec_j \Yvec_j \Avec_j \Bvec_j | E}$ denote the joint distribution of questions and answers of the $j$'th coordinate in this hypothetical strategy, conditioned on the event $E$. Then sampling a tuple $(\Xvec_j, \Yvec_j, \Avec_j, \Bvec_j)$ from $\P_{\Xvec_j \Yvec_j \Avec_j \Bvec_j | E}$ will satisfy the game predicate $V$ with probability $\qval(G) - \eps > \cval(G)$.

Next, we will prove the following three statements (roughly speaking): (1) $\P_{\Xvec_j \Yvec_j | E} \approx \P_{\Xvec_j \Yvec_j}$, (2)~$\P_{\Avec_j | \Xvec_j \Yvec_j E} \approx \P_{\Avec_j | \Xvec_j E}$, and (3) $\P_{\Bvec_j | \Xvec_j \Yvec_j \Avec_j E} \approx \P_{\Bvec_j | \Yvec_j E}$, where ``$\approx$'' denotes closeness in statistical distance. Notice that without the conditioning event $E$, the first item would be trivial and the second item would follow exactly from the non-signaling condition between the players. To prove the third item, we use the fact that the hypothetical strategy for the threshold game uses~$o(n)$ bits of entanglement; intuitively this implies that each instance of $G$ can only use $o(1)$ bits of entanglement. 

Putting these three items together, we obtain a classical strategy for $G$: the first player receives question $\Xvec_j$, and samples an answer $\Avec_j$ from the distribution $\P_{\Avec_j | \Xvec_j E}$. The second player receives question $\Yvec_j$ and samples from $\P_{\Bvec_j | \Yvec_j E}$. The joint distribution of their questions and answers will be close to $\P_{\Xvec_j \Yvec_j \Avec_j \Bvec_j | E}$, but that implies that they will win $G$ with probability $\qval(G) - \eps > \cval(G)$, which is a contradiction.

The proof strategy and the techniques used are heavily inspired by the proofs of the \emph{parallel repetition theorem} in classical complexity theory~\cite{raz1998parallel,Hol09,rao2011parallel}, and subsequently the work on the \emph{quantum parallel repetition problem}. This problem asks for a bound on $\qval(G^n)$ if $\qval(G)<~1$, where $G^n$ is like the threshold game except we demand that the players win \emph{all} instances of~$G$. It is conjectured that $\qval(G^n)$ decays exponentially with $n$, although the best general upper bound is that $\qval(G^n)$ decays polynomially with $n$ when $\qval(G) < 1$~\cite{yuen2016parallel}. Nearly all of the works that study the quantum parallel repetition problem~\cite{JainPY14,chung2015parallel,BVY15fort,BVY17} share the proof strategy of transforming a ``too-good-to-be-true'' strategy for the repeated game~$G^n$ into a ``too-good-to-be-true'' strategy for the single game $G$, namely a quantum strategy with success probability better than $\qval(G)$, a contradiction. These works all use information-theoretic machinery in the proof, and in this work we use the same tools.

\subsection{Related work}

Our work is the first that addresses directly the question of certifying the entanglement of formation of high dimensional states in a noise-tolerant way (while the case of a single CHSH game was already considered in~\cite{verstraete2002entanglement} as mentioned above). 

Any robust self-testing result can be used to certify any continuous entanglement measures (e.g. the entanglement of formation); but as explained before, such results cannot accommodate the kinds of noise considered here. In addition to the self-testing results mentioned before~\cite{mckague2016self,chao2016test,coladangelo2017parallel,coudron2016parallel,natarajan2017quantum,coladangelo2017robust}, the only other self-testing result that certifies asymptotically growing amounts of entanglement is from the work of Reichardt, Unger and Vazirani~\cite{reichardt2013classical}, who show how to verify quantum computations using classical resources only. At the heart of their result is a \emph{sequential} protocol where the experimenter plays many rounds of the CHSH game with the two players in order to certify the presence of many EPR pairs. 
However, like the other self-testing results, the protocol of~\cite{reichardt2013classical} is also not noise-tolerant in the sense considered here. 

If one cares just about certifying high \emph{entanglement rank} of a state  (rather than certifying an entanglement measure such as $E_F$, or precisely characterizing the state as in self-testing), then we can combine the following two independent results to address the question of noise-tolerant, device-independent testing of asymptotically growing amounts of entanglement: 
The work of~\cite{rao2011parallel} shows that the \emph{classical} value of a threshold game $G^n_{\cval(G) + \delta}$ decays exponentially fast with $n$ (if $\cval(G) < 1$). The work of~\cite{junge2010unbounded} shows that the maximum quantum success probability in a game $F$ using dimension-$d$ entanglement is at most $d \, \cval(F)$. Letting $F$ be a threshold game, we obtain that $d$ must be exponentially large in any quantum strategy whose winning probability is say at least a small constant. Since the threshold game is noise-tolerant (i.e. it can be won with high probability with noisy strategies), this gives a noise-tolerant test for \emph{entanglement rank}. This same argument can be modified to show that the \emph{$1/2$-R\'{e}nyi entropy} of the state\footnote{The $1/2$-R\'{e}nyi entropy of a pure state $\ket{\psi}$ is $2 \log (\sum_i \lambda_i^{1/2})$ where $\lambda_i$ are the eigenvalues of the reduced density matrix of $\ket{\psi}$ on either side.} must be linear in $n$.

Our test lower bounds a stronger entanglement measure, the entanglement of formation, which in the pure state case is the entanglement entropy and therefore a lower bound on the $1/2$-R\'{e}nyi entropy. There can be arbitrarily large gaps between the von Neumann entropy and the $1/2$-R\'{e}nyi entropy of a pure state.

The broader goal of certifying the dimension of a quantum system in a device-independent manner has been heavily studied under the heading of \emph{dimension witnesses}. Much of the work on dimension witnesses has focused on finding Bell inequalities such that achieving the optimal violation requires an entangled state of a certain dimension~\cite{brunner2008testing,pal2009quantum,cai2016new}. Many of these works construct and design dimension witnesses using a combination of analytical and numerical techniques.

\subsection{Future work}
Some open problems and future directions include:
\begin{enumerate}
	\item Quantitatively improve our results. The constants $c_1,c_2$ in Theorem~\ref{thm:main} are small; for the CHSH game, the constant $c_1$ is on the order of $10^{-6}$ and thus in order for our Theorem to give any guarantees, $\sim 10^6$ CHSH games would have to be played.
	Even though recent experiments are capable of producing such a large amount of states (in~\cite{liu2017high}, for example, order of $10^{10}$ signals were produced), an improvement of the constants can lead to the ability of certifying \emph{much} more entanglement in such experiments. 
	Our analysis is far from tight and significant quantitative improvements can probably be gained by tailoring the analysis to a specific game, such as the CHSH game.
	
	\item To get a non-trivial bound on the entanglement of formation, this requires that the success probability $\kappa$ is at least $\sim 1/\sqrt{n}$. Can this dependence on $\kappa$ be improved?
	
	\item Can one prove a version of Theorem~\ref{thm:main} for some non-local games $G$ that allows one to lower bound other measures of entanglement, such as distillable entanglement\footnote{In a related work by Jean-Daniel Bancal together with one of the current authors a device-independent protocol certifying a lower bound on the one shot distillable entanglement is given. The considered setting and type of statement are different than the ones presented here. For further details see~\cite{arnon2017noise}.} or quantum conditional entropy? The results of~\cite{vertesi2014disproving,friis2017geometry} indicate that this cannot be done for arbitrary amount of noise for all games since there are Bell inequalities that can be violated while using states with un-distillable entanglement or positive conditional entropy.

	\item Can one prove a self-testing result for a growing number of EPR pairs that is also noise-tolerant in the sense described above? A concrete goal would be to characterize all near-optimal strategies for the threshold game $CHSH^n_{.854 - \nu}$. The results of~\cite{coopmans2017robust} hint that by sticking to the current measures of distance considered in self-testing results any characterization of near-optimal strategies for $CHSH^n_{.854 - \nu}$, in the regime of high amount of noise, must include also non-entangled states. Hence, we do not expect self-testing results (as they are phrased today) to allow for certification of entanglement in the presence of arbitrary noise using threshold games.
\end{enumerate}

\paragraph{Acknowledgments.} We thank Valerio Scarani for helpful pointers to the literature, Thomas Vidick for feedback on an earlier draft, and anonymous referees for helpful comments and pointing us to the work of~\cite{junge2010unbounded}. Work on this project initiated when RAF was visiting UC Berkeley.

\section{Preliminaries}
\label{sec:prelim}

We will use caligraphic font such as $\mathcal{X}$ to denote alphabets. We will use boldfaced font to denote vectors. For example, $\xvec$ will denote an element of $\mathcal{X}^n$. We will use capital boldfaced font to denote the corresponding random variables. For example, $\Xvec$ is a random variable that takes values in $\mathcal{X}^n$. For a coordinate $i$, $\xvec_i$ will denote the $i$'th element of $\xvec$ and $\Xvec_i$ will denote the corresponding random variable. For a subset $S \subseteq [n]$, $\xvec_S$ will denote the sub-tuple of $\xvec$ indexed by $S$.

\subsection{Probability distributions}\label{subsec:prob_dist}
We largely adopt the notational conventions from~\cite{Hol09} for probability distributions. We let capital letters denote random variables and lower case letters denote specific samples. 
We use $\P_X$ to denote the probability distribution of random variable $X$, and $\P_X(x)$ to denote the probability that $X = x$ for some value $x$. For multiple random variables, e.g., $X, Y, Z$, $\P_{XYZ}(x,y,z)$  denotes their joint distribution with respect to some probability space understood from context. 

We use $\P_{Y | X = x}(y)$ to denote the conditional distribution $\P_{YX}(y,x)/\P_X(x)$, which is defined when $\P_X(x) > 0$. When conditioning on many variables, we usually use the shorthand $\P_{X | y,z}$ to denote the distribution $\P_{X | Y =y,Z=z}$. For example, we write $\P_{V | \omega_\mi, x_i, y_i}$ to denote $\P_{V | \Omega_\mi = \omega_\mi, X_i = x_i, Y_i = y_i}$. For an event $W$ we let $\P_{X Y | W}$ denote the distribution conditioned on $W$. We use the notation $\Ex_{x} f(x)$ to denote the expectation $\sum_{x} \P_X(x) f(x)$ when the distribution $\P$ is understood from context.

Let $\P_{X_0}$ be a distribution on $\mathcal{X}$, and for every $x$ in the support of $\P_{X_0}$, let $\P_{Y | X_1 = x}$ be a conditional distribution defined over $\mathcal{Y}$. We define the distribution $\P_{X_0} \P_{Y | X_1}$ over $\mathcal{X} \times \mathcal{Y}$ as
$$
	(\P_{X_0} \P_{Y | X_1})(x,y) \,:=\, \P_{X_0}(x) \cdot \P_{Y | X_1 = x}(y).
$$
Additionally, we write $\P_{X_0 Z} \P_{Y | X_1}$ to denote the distribution $(\P_{X_0 Z} \P_{Y | X_1})(x,z,y) := \P_{X_0 Z}(x,z) \cdot \P_{Y | X_1 = x}(y)$.

For two random variables $X_0$ and $X_1$ over the same set $\X$, we use
$$\| \P_{X_0} - \P_{X_1} \| \,:=\, \frac{1}{2}\sum_{x \in \mathcal{X}} |\P_{X_0}(x) - \P_{X_1} (x)|,$$
to denote the total variation distance between $\P_{X_0}$ and $\P_{X_1}$. We will use the shorthand $\P_{X_0} \approx_\delta \P_{X_1}$ to denote $\| \P_{X_0} - \P_{X_1} \| \leq \delta$.

Additionally, given two probability distributions $\P_{XY}, \Q_{XY}$ such that $\P_X = \Q_Y$ (i.e. the marginals are the same) we will write expressions such as
\[
	\Ex_{x} \Brac{ \P_{Y | x} \approx_\delta \Q_{Y | x} }
\]
to denote 
\[
	\Ex_{x} \Norm{ \P_{Y | x} - \Q_{Y | x} } \leq \delta
\]
where the expectation over $x$ drawn from $\P_X$.

\subsection{Quantum information theory}

For comprehensive references on quantum information we refer the reader to~\cite{nielsen2010quantum,wilde2013quantum}. 

For a matrix $A$, we will use $\| A \|_1$ to denote its \emph{trace norm} $\Tr(\sqrt{A A^\dagger})$. A density matrix is a positive semidefinite matrix with trace $1$. For Hermitian matrices $A, B$ we write $A \preceq B$ to indicate that $A - B$ is positive semidefinite. We use $\Id$ to denote the identity matrix. A \emph{positive operator valued measurement} (POVM) with outcome set $\mathcal{A}$ is a set of positive semidefinite matrices $\{ E^a \}$ labeled by $a\in \mathcal{A}$ that sum to the identity.

We use sans-serif font such as $\A, \B, \X, \Y$ to denote system labels. We will decorate quantum states with superscripts to denote the relevant registers; so $\rho^{\A\B}$ will denote the density matrix on the systems $\A$ and $\B$. We will let $\density{\A}$ to denote the set of density matrices on system~$\A$. A \emph{classical-quantum} state (or simply \emph{cq-state}) $\rho^{\X \mathsf{E}}$ is classical on $\X$ and quantum on $\mathsf{E}$ if it can be written as $\rho^{\X\mathsf{E}} = \sum_{x} p(x) \ketbra{x}{x}^\X \otimes \rho^{\mathsf{E}}_{X = x}$ for some probability measure $p(\cdot)$. For notational convenience, we will use $\dbrac{x}$ to denote the classical register $\ketbra{x}{x}$. 

The state $\rho^{\mathsf{E}}_{X = x}$ is by definition the $\mathsf{E}$ part of the state $\rho^{\X\mathsf{E}}$, conditioned on the classical random variable $X = x$. We write $\rho^{\X\mathsf{E}}_{X = x}$ to denote the state $\dbrac{x}^{\X} \otimes \rho^{\mathsf{E}}_{X = x}$. We often write expressions such as $\rho^{\mathsf{E}}_x$ as shorthand for $\rho^{\mathsf{E}}_{X = x}$ when it is clear from context which registers are being conditioned on. This will be useful when there are many classical variables to be conditioned on. 


We will use the short hand $\rho \approx_\delta \sigma$ to denote $\norm{\rho - \sigma}_1 \leq \delta$. We use the expression $ \Ex_{z} \brac{ \rho_z \approx_\delta \sigma_z}$ to denote $\Ex_z \norm{\rho_z - \sigma_z}_1 \leq \delta$.


\paragraph{Relative entropy, relative min-entropy, and mutual information.} For two positive semidefinite operators $\rho$, $\sigma$, the \emph{relative entropy} $D(\rho \| \sigma)$ is defined to be $\Tr(\rho (\log \rho - \log \sigma))$. The \emph{relative min-entropy} $D_\infty(\rho \| \sigma)$ is defined as $\min\{ \lambda : \rho \preceq 2^\lambda \sigma \}$.   

Let $\rho^{\A\B}$ be a bipartite state. The \emph{mutual information} $I(\A:\B)_\rho$ is defined as $D(\rho^{\A\B} \| \rho^\A \otimes \rho^\B)$. For a classical-quantum state $\rho^{\X\A\B}$ that is classical on $X$ and quantum on $\A\B$, we write $I(\A : \B | x)_{\rho}$ to indicate $I(\A : \B)_{\rho_x}$.

\begin{lemma}[Pinsker's inequality]
\label{lemma:pinsker}
	For all density matrices $\rho, \sigma$, $ \frac{1}{2} \| \rho - \sigma \|^2_1 \leq D(\rho \| \sigma)$.
\end{lemma}

\begin{lemma}
\label{lem:div_cond}
For density matrices $\rho$, $\sigma$ such that $\rho \preceq 2^K \sigma$ in the positive semidefinite ordering, we have that $D(\rho \, \| \, \sigma) \leq K$.
\end{lemma}

\begin{lemma}[\cite{JainPY14}, Fact II.8]
\label{lem:div_chain_rule}
	Let $\P_Z$ and $\Q_Z$ be distributions. Let $\rho = \Ex_{z \sim \P_Z} \dbrac{z} \otimes \rho_z$, and $\rho' = \Ex_{z \sim \Q_Z}  \dbrac{z} \otimes \rho'_z$. Then $D(\rho' \| \rho) = D(\Q_Z \| \P_Z) + \Ex_{z \sim \Q_Z} \left [ D(\rho'_z \| \rho_z) \right]$. In particular, $D(\rho' \| \rho) \geq \Ex_{Z \sim \Q_Z} \left [ D(\rho'_z \| \rho_z) \right]$.
\end{lemma}
We will also use the following Lemma from \cite{chung2015parallel,BVY17}.

\begin{lemma}[\cite{chung2015parallel,BVY17}, Quantum Raz's Lemma] \label{lem:quantum_raz} 
Let $\rho$ and $\sigma$ be two CQ states with  $\rho^{\X\A}= \rho^{\X_1 \X_2 \ldots \X_n A}$ and $\sigma= \sigma^{\X\A}= \sigma^{\X_1}\otimes \sigma^{\X_2}\otimes \ldots \otimes \sigma^{\X_n} \otimes \sigma^A$ with $\X=\X_1 \X_2 \ldots \X_n$ classical in both states. Then
\begin{equation}\label{eqn:Raz_lemma1} \sum_{i=1}^n I(\X_i \, :\, \A)_\rho \leq D(\rho^{\X\A} \, \| \sigma^{\X\A}). \end{equation}
\end{lemma}

\paragraph{Randomized chain rule.} The standard chain rule for mutual information states that for an $n$-partite system $\X_1,\ldots,\X_n$, we have that $I(\X_1 \cdots \X_n : \A) = \sum_i I(\X_i : \A | \X_{< i})$. However, there are many ways of performing the chain rule, depending on the ordering of the $\X_i$'s. It is useful to average over many possible ways of performing the chain rule: 
\[
	\Ex_\pi \sum_i I(\X_{\pi(i)} : \A | \X_{\pi(< i)}) = I(\X_1 \cdots \X_n : \A).
\]
We call this the ``randomized chain rule.'' Here $\pi$ is a uniformly random permutation on $n$ elements, and $\pi(< i)$ denotes the image of the permutation applied to $\{1,\ldots,i-1\}$.

\subsection{Entanglement measures}
There are many different way of quantifying the entanglement of a bipartite quantum state~\cite{plenio2005introduction,horodecki2009quantum}. We will use three of them in the current work.

\begin{definition}\label{def:entang_entropy}
	For a pure state $\ket{\psi}^{\Q_A\Q_B}$, the entanglement entropy is
	\[
		E(\psi) = H \left(\Q_B\right)_{\psi} \;,
	\]
	where $H$ is the von Neumann entropy.
\end{definition}

\begin{definition}
	For a mixed state $\rho^{\Q_A\Q_B}$, the entanglement of formation is
	\[
		E_F(\rho^{\Q_A\Q_B}) = \inf \left\{ \sum_t p_t E(\psi) \;\; : \;\; \rho=\sum_t p_t \ket{\psi_t}\bra{\psi_t} \right\} \;,
	\]
	where $E(\psi)$ is the entanglement entropy of $\psi$, as in Definition~\ref{def:entang_entropy}.
\end{definition}

\begin{definition}
	For a mixed state $\rho^{\Q_A\Q_B}$, the entanglement cost is
	\begin{equation*}
		E_C(\rho^{\Q_A\Q_B}) = \inf \left\{ r : \lim_{n \rightarrow \infty} \left( \inf_{\Lambda} \| \rho^{\otimes n} - \Lambda(\Phi^+_{2^{rn}}) \|_1 \right) =0 \right\} \;,
	\end{equation*}
	where the infimum ranges over all LOCC maps $\Lambda$ and $\Phi^+_{2^{rn}}$ is the maximally entangled state of rank~$2^{rn}$ (both with respect to the partition $\Q_A$ vs.\@ $\Q_B$). 
\end{definition}

\subsection{Classical correlated sampling}

\emph{Correlated sampling} is a key component of Holenstein's proof of the classical parallel repetition theorem.

\begin{lemma}[Classical correlated sampling~\cite{Hol09}]
\label{lem:ccorsamp}
	Let $\P$ and $\Q$ be two probability distributions over a universe $\mathcal{U}$ such that $\| \P - \Q \|_1 \leq \eps < 1$. Then there exists a zero communication two-player protocol using shared randomness where the first player outputs an element $p \in \mathcal{U}$ distributed according to $\P$, the second player samples an element $q \in \mathcal{U}$ distributed according to $\Q$, and with probability at least $1 - 2\eps$, the two elements are identical (i.e. $p = q$).
\end{lemma}
\noindent We call the protocol in the Lemma above the \emph{classical correlated sampling procedure}.

\subsection{Two-player games}\label{sec:games_related_def}

A two-player game $G$ is a tuple $(\mu,V)$ where $\mu$ is a question distribution over some alphabet $\mathcal{X} \times \mathcal{Y}$ and $V: \X \times \Y \times \mathcal{A} \times \mathcal{B} \to \{0,1\}$ is a verification predicate with $\mathcal{A}$ and $\mathcal{B}$ denoting the answer alphabets for Alice and Bob respectively. Operationally, in a two-player game, a referee samples a question pair $(x,y)$ from $\mu$, and sends $x$ to Alice and $y$ to Bob. Alice responds with answer $a$, Bob responds with answer $b$, and the referee decides to accept or reject based on the predicate $V(x,y,a,b)$.

A quantum strategy for $G$ consists of 
\begin{itemize}
	\item A shared entangled state $\rho \in  \density{\Q_A \otimes \Q_B}$ where $\Q_A, \Q_B$ are Hilbert spaces isomorphic to $\C^d$ for some finite $d$.
	\item Measurement elements $\{ A_x(a) \}$, $\{ B_y(b) \}$ acting on $\Q_A$ and $\Q_B$ respectively. By measurement elements we mean that for every $x$, $\sum_a A_x(a)^2 = \Id$, and similarly for~$B$. 
\end{itemize}
A strategy is a \emph{pure state strategy} if the shared state $\rho$ is rank one. 

The \emph{quantum value} of a game $G$, denoted $\qval(G)$, is the maximum success probability in $G$ over all (finite-dimensional) quantum strategies:
\[
	\qval(G) := \max_{ (\rho, \{A_x(a)\}, \{B_y(b) \}) } \sum_{x,y} \mu(x,y) \sum_{\substack{a, b : \\ V(x,y,a,b) = 1}} \Tr \left ( A_x(a) \otimes B_y(b) \rho \right).
\]
The \emph{classical value} of a game $G$, denoted $\cval(G)$, is the maximum success probability in $G$ over classical strategies (strategies where $\rho$ is separable across $\Q_A$ and $\Q_B$). 


\paragraph{Threshold games.} The threshold game $G^n_{1 - \gamma}$ is a game where the question distribution is $\mu^n$, and the verification predicate $V^n_{1 - \gamma}: \mathcal{X}^n \times \mathcal{Y}^n \times \mathcal{A}^n \times \mathcal{B}^n \to \set{0,1}$ is such that $V^n_{1 - \gamma}(\xvec,\yvec,\avec,\bvec) = 1$ if and only if at least an $1 - \gamma$ fraction of coordinates $i \in [n]$ are won, i.e., $V(\xvec_i,\yvec_i,\avec_i,\bvec_i) = 1$. 

%
%

\paragraph{Probability distribution $\P$.} We will refer to a probability distribution $\P$ on random variables $\Xvec, \Yvec, \Avec, \Bvec$ which correspond to the questions and answers of Alice and Bob, respectively in the game $G^n_\alpha$ played according to the strategy above. More precisely, $\P_{\Xvec \Yvec}(\xvec,\yvec)$ is the distribution of questions in $G^n_\alpha$, or $\mu^n$. Then, for every question pair $(\xvec,\yvec)$,
\begin{equation}\label{eq:prob_dist_from_state}
	\P_{\Avec \Bvec | \xvec \yvec}(\avec,\bvec) = \Tr( A_\xvec(\avec) \otimes B_\yvec(\bvec) \rho)
\end{equation}
describes a given strategy. Thus, the full joint distribution $\P_{\Xvec \Yvec \Avec \Bvec}$ is 
\[
	\P_{\Xvec \Yvec \Avec \Bvec} = \P_{\Xvec \Yvec} \cdot \P_{\Avec \Bvec | \Xvec \Yvec}.
\]

When considering marginals of $\P_{\Xvec \Yvec\Avec\Bvec}$ it is understood that an expectation is taken over all the registers which are not explicitly stated. For example, for $i\in[n]$ we can write $ \P_{\Avec_i\Bvec_i}(\avec_i,\bvec_i) = \Ex_{\xvec} \Ex_{\yvec} \sum_{\avec,\bvec | \avec_i,\bvec_i} \P_{\Avec \Bvec | \xvec \yvec}(\avec,\bvec)$.

\paragraph{Dependency-breaking variable.} Fix a subset $S \subseteq [n]$. We will define \emph{dependency-breaking variables} as follows. Let $D_1,\ldots,D_n$ be independent and uniformly distributed in $\{Alice,Bob\}$. Let $M_1,\ldots,M_n$ be independent random variables defined in the following way: for each $i \in [n]$,
\begin{align*}
	M_i = \left \{ \begin{array}{ll}
		\Xvec_i & \mbox{ if } D_i = Alice \\
		\Yvec_i & \mbox{ if } D_i = Bob
		\end{array}
	\right.
\end{align*}
Now for $i \in [n]$, we define $\Omega_i := (D_i,M_i)$. We say that $\Omega_i$ \emph{fixes Alice's input} if $D_i = Alice$, and otherwise $\Omega_i$ fixes Bob's input. We write $\Omega$ to denote the random variable $(\Omega_1,\ldots,\Omega_n,\Xvec_S,\Yvec_S)$, where $\Xvec_S\Yvec_S$ are Alice and Bob's questions in the coordinates indexed by $S$. For $i \in [n]$ we write $\Omega_\mi$ to denote the random variable $\Omega$ with $\Omega_i$ omitted.

We will augment the probability space $\P$ of $(\Xvec,\Yvec,\Avec,\Bvec)$ with random variables $\Omega$, $Z$ to obtain the joint distribution $\P_{\Omega \Xvec \Yvec \Avec \Bvec}$.

\begin{claim}
\label{clm:dependency_breaking}
For every fixing of $\Omega = \omega$, we have
\[
	\P_{\Xvec \Yvec | \omega} = \P_{\Xvec | \omega} \cdot \P_{\Yvec | \omega}.
\]
That is, conditioned on $\Omega$, $\Xvec$ and $\Yvec$ are independent random variables.
\end{claim}

This notion of dependency-breaking variables comes from proofs of parallel repetition theorems and communication complexity lower bounds in \emph{classical} theoretical computer science~\cite{raz1998parallel, bar2002information}.

\section{Proof of Theorem~\ref{thm:main}}
\label{sec:proof}

The main technical part of our work is proving Theorem~\ref{thm:main_pure} given below. The statement of Theorem~\ref{thm:main_pure} is almost identical to that of Theorem~\ref{thm:main}, but it is restricted to testing the entanglement and dimension of \emph{pure} states (i.e., the state $\rho$ shared by the two players is a rank one density matrix $\ketbra{\psi}{\psi}$).

\begin{theorem}[Main theorem for pure state strategies]\label{thm:main_pure}
	Let $G$ be a game with classical-quatum gap $\Delta = \qval(G) - \cval(G) > 0$. Let $0 < \nu \leq \Delta$. 
	Denote
	\begin{equation}\label{eq:constants_expl}
		c_1' = \frac{(\Delta - \nu)^3}{2000 \cdot \log |\mathcal{A} \times \mathcal{B}|}  \quad \;\text{and} \quad c_2' =  \frac{(\Delta - \nu)^5}{10\cdot 90^2 \cdot \log |\mathcal{A} \times \mathcal{B}|} \;,
	\end{equation}
	where $\mathcal{A}$ and $\mathcal{B}$ are the answer alphabets in $G$ for Alice and Bob, respectively. 
	
	For all integer $n$ greater than $\frac{1}{c_1'}$, any pure state strategy that wins the threshold game $G_{\qval(G) - \nu}^n$ with probability $\kappa \geq \exp(-c_1' n)$ involves an entangled state $\ket{\psi}$ satisfying $E(\psi) \geq c_2' \kappa n$.
\end{theorem}

As we now show, Theorem~\ref{thm:main_pure} implies Theorem~\ref{thm:main} with a slight loss in the constants~$c_1$ and ~$c_2$.

\begin{proof}[Proof of Theorem~\ref{thm:main} from Theorem~\ref{thm:main_pure}]

The completeness portion of Theorem~\ref{thm:main} follows straightforwardly from Hoeffding's bound. We will now concentrate on the soundness part of the Theorem. 

We set $c_1 = 2c_1'$ and $c_2 = c_2'/4$ where $c_1'$ and $c_2'$ are the constants given in Equation~\eqref{eq:constants_expl}. 

Consider a mixed state strategy for $G_{\qval(G) - \nu}^n$ that uses $\rho \in \density{\Q_A \otimes \Q_B}$ as the shared state between Alice and Bob, and has success probability $\kappa > \exp(-c_1 n)$. 
	
	Let $\rho = \sum_i p_i \ketbra{\psi_i}{\psi_i}$ be a decomposition of $\rho$ that realizes the entanglement of formation $E_F(\rho)$. That is, we have
	\[
		E_F(\rho) = \sum_i p_i E(\psi_i).
	\]
	Let $\kappa_i$ denote the probability that the players win $G_{\qval(G) - \nu}^n$ if instead of using the mixed state $\rho$ they used $\ket{\psi_i}$ (but used the same measurement operators). We have then that $\kappa = \sum_i p_i \kappa_i$. 
	
	If we sample $i$ according to $p_i$, then with probability at least $\kappa/2$ we have $\kappa_i \geq \kappa/2$, according to Markov's inequality.  Take such an $i$. Then $\kappa_i \geq \kappa/2 \geq \exp(-c_1 n)/2 > \exp(-c_1' n)$. Invoking Theorem~\ref{thm:main_pure} for pure states, we get that the entanglement of formation $E_F(\psi_i) = E(\psi_i) \geq c_2' \kappa_i n$.
	Therefore,
	\[
		E_F(\rho) \geq \sum_{i:\kappa_i \geq \kappa/2} p_i E(\psi_i) \geq \Paren{\sum_{i:\kappa_i \geq \kappa/2} p_i } c_2' \kappa_i n \geq \frac{\kappa}{2} \cdot c_2' \kappa_i n \geq c_2 \kappa^2 n
	\]
	establishing the Soundness condition of Theorem~\ref{thm:main}. 
		
	The final constants are given by 
	\begin{equation}\label{eq:final_thm_constants}
		c_1 = \frac{(\Delta - \nu)^3}{1000 \cdot \log |\mathcal{A} \times \mathcal{B}|} \quad \text{and} \quad c_2 = \frac{(\Delta - \nu)^5}{10 \cdot 180^2 \cdot \log |\mathcal{A} \times \mathcal{B}|} \;.
	\end{equation}
	Hence, for any amount of noise $\nu$ which does not result in threshold games which is ``effectively classical'', $c_1,c_2>0$ and Theorem~\ref{thm:main} is non-trivial.
\end{proof}

\begin{proof}[Proof of Theorem~\ref{thm:main_pure}]
The probability distribution $\P$ is defined relative to the hypothesized pure state strategy as in Equation~\eqref{eq:prob_dist_from_state}. Let $1 - \gamma = \qval(G) - \nu$ and let $1 - \eps =\cval(G)$ (so therefore $\eps - \gamma = \Delta - \nu$). 

As alluded to in the introduction, we will attempt to construct a ``too-good'' \emph{classical} strategy for $G$ by simulating playing a random coordinate $j$ of $G_{1 - \gamma}^n$, conditioned on a particular event. The main point is that if too little entanglement is used for $G_{1 - \gamma}^n$, then this simulation can be performed without any entanglement at all.

The following Proposition identifies what this conditioning event is. It shows that there exists a small subset of coordinates $S$ such that, conditioned on winning more than $1 - \tau$ fraction of $S$, the probability of winning a random coordinate $j$ outside of $S$ is high. Define the following events:
\begin{itemize}
	\item $W_j$ denotes the event that the players win the $j$'th coordinate.
	\item $W^{\geq 1 - \gamma}$ denotes the event that the players win more than $(1 - \gamma)n$ games. 
	\item $\theevent$ denotes the event that the players win more than $1 - \tau$ fraction of coordinates in $S$.
\end{itemize}

\begin{proposition}\label{prop:subset2}
	Let $\alpha = \eps - \gamma$. Let $\tau = \eps - \frac{3}{4}\alpha $. Suppose that $\P(W^{\geq 1 - \gamma}) \geq \frac{16}{\alpha} 2^{-\alpha^3 n/384}$. Then there exists a set $S \subseteq [n]$ of size at most $\frac{96}{\alpha^2} \Paren{ \ln \frac{16}{\alpha \P(W^{\geq 1 - \gamma})}}$ such that
	$$
		\Ex_{j \notin S} \P(W_j | \theevent) \geq 1 - \eps + \alpha.
	$$
	where $j$ is chosen uniformly from $[n] - S$, and $\P(W_S^{\geq 1 - \tau}) \geq \P(W^{\geq 1 - \gamma})/2$.
\end{proposition}
We defer the proof of Proposition~\ref{prop:subset2} to the Appendix. 

Fix a set $S$ given by the Proposition for the rest of the proof, and the dependency-breaking variables $\Omega$ (as defined in Section~\ref{sec:games_related_def}) and $\Avec_S,\Bvec_S$ will be defined relative to this $S$. From Proposition~\ref{prop:subset2} we know that $\P(\theevent) \geq \kappa/2$.

Without loss of generality, let us number the coordinates so that $S = \{ n - |S| + 1,\ldots, n \}$, and set $m = n - |S|$.

\medskip
\medskip

Having identified the special conditioning event $\theevent$, we will present four main Lemmas upon which the proof rests. For every set $T \subseteq [m]$ and coordinate $j \notin T \cup S $, let
\[
	R_{Tj} = (\Omega_\mj, \Xvec_T, \Avec_\ScupT,\Bvec_S) \;,
\]
that is, the dependency-breaking variable $\Omega$ with the $j$'th coordinate omitted, Alice's questions in the set $T$, Alice's answers in the set $S \cup T$, and Bob's answers in the set $S$.

Let $0 < \beta < 1$ be a parameter that we will set later. Define the error parameters
\begin{itemize}
	\item $\delta := \frac{1}{(1 - \beta)m} \log \frac{1}{\P(\theevent)}$.
	\item $\delta' := \frac{1}{(1 - \beta)m} \Paren{\log \frac{1}{\P(\theevent)} + (2|S| + \beta m) \log |\mathcal{A} \times \mathcal{B}|}$.
	\item $\delta'' := \frac{1}{\beta m} \frac{E(\psi)}{\P(\theevent)}$
\end{itemize}
where $E(\psi)$ is the entanglement entropy of $\ket{\psi}$ and $D$ denotes the Schmidt rank of $\ket{\psi}$, where $\ket{\psi}$ is the state used in the strategy. 

At a high level, the proof proceeds as follows: 
Proposition~\ref{prop:subset2} implies that for some $j$, the distribution $\P_{\Xvec_j \Yvec_j \Avec_j \Yvec_j | \theevent}$ gives rise to questions and answer tuples that win $G$ with ``too good'' probability. We can split this distribution as 
\[\P_{\Xvec_j \Yvec_j | \theevent} \cdot \P_{\Avec_j | \Xvec_j \Yvec_j \theevent} \cdot \P_{\Bvec_j | \Xvec_j \Yvec_j \Avec_j \theevent}. \]
Suppose the following approximations were established:
\begin{align*}
	&\P_{\Xvec_j \Yvec_j | \theevent} \approx \P_{\Xvec_j \Yvec_j} \\
	&\P_{\Avec_j | \Xvec_j \Yvec_j \theevent} \approx \P_{\Avec_j | \Xvec_j \theevent}\\
	&\P_{\Bvec_j | \Xvec_j \Yvec_j \Avec_j \theevent} \approx \P_{\Bvec_j | \Yvec_j  \theevent}.
\end{align*}
We would be done, because this would mean that Alice and Bob could sample the correct distribution of answers without having to know the other person's question, and thus sample answers that win with ``too good'' probability.

Lemmas~\ref{lem:input_dist}-\ref{lem:alice_ans} stated below imply that there is a random variable $R$ that is (a) jointly sampleable by Alice and Bob and (b) the approximations above hold when conditioned on $R$. 
In the lemmas, the expression $\Ex_{j \notin T\cup S}$ denotes a uniformly random index $j$ not in the set $T\cup S$.

\begin{lemma}[Input distribution is unchanged]
\label{lem:input_dist}
For every set $T$ of size at most $\beta m$, 
	\[
		\Ex_{j \notin T \cup S} \,\, \Brac{\P_{\Xvec_j \Yvec_j | \theevent} \approx_{\sqrt{\delta}} \P_{\Xvec_j \Yvec_j}}.
	\]
\end{lemma}

\begin{lemma}[Dependency-breaking variable is correlatedly sampleable]
For every set $T$ of size at most $\beta m$,
\label{lem:corr_samp}
	\[
		\Ex_{j \notin T \cup S} \,\,\, \Ex_{\xvec_j,\yvec_j | \theevent}  \Brac{\P_{R_{Tj} | \xvec_j, \yvec_j, \theevent} \approx_{\sqrt{\delta'}} \P_{R_{Tj} | \xvec_j, \theevent} \approx_{\sqrt{\delta'}} \P_{R_{Tj} | \yvec_j, \theevent}}.
	\]
\end{lemma}

\begin{lemma}[Bob's answer does not depend on Alice's answer and question]
\label{lem:bob_ans}
	\[
		\Ex_{T,j \notin T \cup S} \,\,\, \Ex_{\xvec_j,\yvec_j,r_{Tj},\avec_j | \theevent} \Brac{\P_{\Bvec_j | r_{Tj}, \xvec_j, \avec_j, \yvec_j, \theevent} \approx_{\sqrt{2\delta''}} \P_{\Bvec_j | r_{Tj}, \yvec_j, \theevent} }
	\]
	where $T$ is a uniformly random set size at most $\beta m$.
\end{lemma}

\begin{lemma}[Alice's answer does not depend on Bob's question]
For every set $T$ of size at most $\beta m$,
\label{lem:alice_ans}
	\[
		\Ex_{j \notin T \cup S} \,\,\, \Ex_{\xvec_j,\yvec_j,r_{Tj} | \theevent} \Brac{\P_{\Avec_j | r_{Tj}, \xvec_j, \yvec_j, \theevent} \approx_{\sqrt{2\delta'}} \P_{\Avec_j | r_{Tj}, \xvec_j, \theevent} }.
	\]
\end{lemma}

The proofs of the lemmas are given in Section~\ref{sec:proofs_lemmas}.

As explained above, by putting Proposition~\ref{prop:subset2} together with these four Lemmas we can achieve our goal of simulating a random coordinate of $G^n_{1 - \gamma}$ classically. 
We now make this precise. Consider the following protocol to play game $G$:

\begin{figure}[H]
\begin{center}
\textbf{Protocol for game $G$} \\
\medskip
\framebox{
\begin{minipage}{0.9\textwidth}
Alice and Bob receive input $(x,y)$ sampled according to $\mu$.
\begin{enumerate}
	\item Alice and Bob use shared randomness to jointly sample a uniformly random set $T$ of size at most $\beta n$, and an index $j\in [m] \setminus T$ uniformly at random. 
	\item Alice sets $\xvec_j \leftarrow x$, Bob sets $\yvec_j \leftarrow y$.
	\item Alice and Bob use correlated sampling (Lemma~\ref{lem:ccorsamp}) to jointly sample the random variable $R_{Tj}$. Alice obtains a sample $r_{Tj}^A$ distributed according to $\P_{R_{Tj} | \xvec_j, \theevent }$, Bob obtains a sample $r_{Tj}^B$ distributed according to $\P_{R_{Tj} | \yvec_j, \theevent }$.
	\item Alice outputs a sample $\avec_j$ from the distribution $\P_{\Avec_j | r_{Tj}^A, \xvec_j,  \theevent}$.
	\item Bob outputs a sample $\bvec_j$ from the distribution $\P_{\Bvec_j | r_{Tj}^B, \yvec_j,  \theevent}$
\end{enumerate}

\end{minipage}
}

\end{center}
\end{figure}

Let $\wt{\P}_{R_{Tj}^A R_{Tj}^B | \xvec_j, \yvec_j}$ denote the distribution of $r_{Tj}^A,r_{Tj}^B$ as sampled in the protocol, given the players' inputs. 
Let $\P_{\Avec_j | \xvec_j, r_{Tj} \theevent}$ and $\P_{\Bvec_j | \yvec_j, r_{Tj} \theevent}$ denote the distributions of Alice's and Bob's answers in Steps 4 and 5, respectively. The joint distribution of the random variables $X, Y, R_{Tj}, \Avec_j, \Bvec_j$ in the protocol, averaged over the players' choices of $T,j$ 
is:
\begin{align}
	&\Ex_{T,j} \P_{X Y} \cdot \wt{\P}_{R_{Tj}^A R_{Tj}^B | \xvec_j, \yvec_j} \cdot \wt{\P}_{\Avec_j | \xvec_j, r_{Tj}^A} \cdot \wt{\P}_{\Bvec_j | \yvec_j, r_{Tj}^B} \\
	= &\Ex_{T,j} \P_{\Xvec_j \Yvec_j} \cdot \wt{\P}_{R_{Tj}^A R_{Tj}^B | \xvec_j, \yvec_j} \cdot \wt{\P}_{\Avec_j | \xvec_j, r_{Tj}^A} \cdot \wt{\P}_{\Bvec_j | \yvec_j, r_{Tj}^B} \label{eq:prot2} \\
	\approx_{\sqrt{\delta}} &\Ex_{T,j} \P_{\Xvec_j \Yvec_j | \theevent} \cdot \wt{\P}_{R_{Tj}^A R_{Tj}^B | \xvec_j, \yvec_j} \cdot \wt{\P}_{\Avec_j | \xvec_j, r_{Tj}^A} \cdot \wt{\P}_{\Bvec_j | \yvec_j, r_{Tj}^B} \label{eq:prot3}
\end{align}
In~\eqref{eq:prot2}, we identified $\Xvec_j,\Yvec_j$ with $X,Y$, respectively. In~\eqref{eq:prot3}, we used Lemma~\ref{lem:input_dist}. 

This next claim will allow us to approximate $\wt{\P}_{R_{Tj}^A R_{Tj}^B | \xvec_j, \yvec_j}$ with $\P_{R_{Tj} | \xvec_j, \yvec_j, \theevent}$. 
\begin{claim}
Let $F$ denote the event that $R_{Tj}^A = R_{Tj}^B$. Then the following two approximations hold:
\begin{enumerate}
	\item  $\Ex_{T,j} \Brac{\P_{\Xvec_j \Yvec_j | \theevent} \cdot \wt{\P}_{R_{Tj}^A R_{Tj}^B | \xvec_j, \yvec_j} \approx_{2\sqrt{\delta'}} \P_{\Xvec_j \Yvec_j | \theevent} \cdot \wt{\P}_{R_{Tj}^A R_{Tj}^B | \xvec_j, \yvec_j, F}} \;;$ 
	\item  $\Ex_{T,j} \Brac{\P_{\Xvec_j \Yvec_j | \theevent} \cdot \wt{\P}_{R_{Tj}^A R_{Tj}^B | \xvec_j, \yvec_j, F} \approx_{4\sqrt{\delta'}} \P_{\Xvec_j \Yvec_j | \theevent} \cdot \P_{R_{Tj}| \xvec_j, \yvec_j, \theevent}}\;,$
\end{enumerate}
where in the second approximation, $\P_{R_{Tj}| \xvec_j, \yvec_j, \theevent}$ can be formally understood as $\P_{R_{Tj}R_{Tj}| \xvec_j, \yvec_j, \theevent}$.
\end{claim}
\begin{proof}
 By Lemma~\ref{lem:ccorsamp}, this probability of the event $F$ conditioned on $\xvec_j,\yvec_j,T,j$ is at least $1 - 2\lambda_{\xvec_j,\yvec_j,T,j}$, where
\[
	\lambda_{\xvec_j,\yvec_j,T,j} := \| \P_{R_{Tj} | \xvec_j, \theevent } - \P_{R_{Tj} | \yvec_j, \theevent } \|.
\]
By Lemma~\ref{lem:corr_samp} we have $\Ex_{T,j} \,\,\, \Ex_{\xvec_j,\yvec_j | \theevent} \lambda_{\xvec_j,\yvec_j,T,j} \leq \sqrt{\delta'}$. 
Combined with the fact that the statistical distance between a distribution $D$ and $D$ conditioned on an event of probability $1 - \delta$ is at most $2\delta$, we obtain the first item of the claim. By definition of the protocol, Alice's sample $r_{Tj}^A$ is distributed according to $\P_{R_{Tj} | \xvec_j, \theevent}$, and so therefore $\wt{\P}_{R_{Tj}^A | \xvec_j \yvec_j, F} \approx_{4\lambda_{\xvec_j,\yvec_j,T,j}} \P_{R_{Tj} | \xvec_j, \theevent}$, implying the second item of the claim.
\end{proof}

\medskip
\medskip
\noindent We can now continue approximating line~\eqref{eq:prot3}:
\begin{align}
	\text{\eqref{eq:prot3}} \approx_{6\sqrt{\delta'}} &\Ex_{T,j} \P_{\Xvec_j \Yvec_j | \theevent} \cdot \P_{R_{Tj} | \xvec_j, \yvec_j, \theevent} \cdot \P_{\Avec_j | \xvec_j, r_{Tj} \theevent} \cdot \P_{\Bvec_j | \yvec_j, r_{Tj} \theevent} \label{eq:prot4} \\
	= &\Ex_{T,j} \P_{R_{Tj} \Xvec_j \Yvec_j | \theevent} \cdot \P_{\Avec_j | \xvec_j, r_{Tj}, \theevent} \cdot \P_{\Bvec_j | \yvec_j, r_{Tj} \theevent} \label{eq:prot5} \\
	\approx_{\sqrt{2\delta'}} &\Ex_{T,j} \P_{R_{Tj} \Xvec_j \Yvec_j | \theevent} \cdot \P_{\Avec_j | \xvec_j, \yvec_j, r_{Tj}, \theevent} \cdot \P_{\Bvec_j | \yvec_j, r_{Tj} \theevent} \label{eq:prot6} \\
	\approx_{\sqrt{2\delta''}} &\Ex_{T,j} \P_{R_{Tj} \Xvec_j \Yvec_j | \theevent} \cdot \P_{\Avec_j | \xvec_j, \yvec_j,  r_{Tj}, \theevent} \cdot \P_{\Bvec_j | r_{Tj}, \xvec_j,\avec_j, \yvec_j, \theevent} \label{eq:prot7} \\
	= &\Ex_{T,j} \P_{R_{Tj} \Xvec_j \Yvec_j \Avec_j \Bvec_j | \theevent} \label{eq:prot8}
\end{align}
In~\eqref{eq:prot4}, we used the claim just proven. In~\eqref{eq:prot5}, we used the definition of how Alice samples $\Avec_j$. In~\eqref{eq:prot6}, we used Lemma~\ref{lem:alice_ans}. In~\eqref{eq:prot7}, we used Lemma~\ref{lem:bob_ans}. 

Consider the marginal distribution of $(\Xvec_j,\Yvec_j,\Avec_j,\Bvec_j)$ in $\Ex_{T,j} \P_{R_{Tj}, \Xvec_j \Yvec_j, \Avec_j, \Bvec_j | \theevent}$; this is simply $\Ex_{j \notin S} \P_{\Xvec_j \Yvec_j, \Avec_j, \Bvec_j | \theevent}$. Setting $\alpha = \eps - \gamma$, by Proposition~\ref{prop:subset2}, the probability that $(\Xvec_j,\Yvec_j,\Avec_j,\Bvec_j)$ satisfies the game $G$ predicate is at least $1 - \eps + \alpha$, and therefore the probability the same is true in the protocol will be at least $1 - \eps + \alpha - (\sqrt{\delta} + 8\sqrt{\delta'} + \sqrt{2\delta''})$, because of the errors accrued in the approximations above. 

Let 
\[
	C = \log |\mathcal{A} \times \mathcal{B}|, \qquad \beta = \frac{\alpha^2}{1000 \cdot C} \;.
\] 
Using the assumptions that $\kappa \geq 2^{- \alpha^3 n/1000 C}$, we get that both $\sqrt{\delta} \leq \alpha/3$ and $8\sqrt{\delta'} \leq 4\alpha/9$. Thus if \[
\frac{2 E(\psi) }{\kappa} < \frac{2 \alpha^2 \beta m}{81}\;,
\]
then we have $\alpha > \sqrt{\delta} + 8\sqrt{\delta'} + \sqrt{2\delta''}$, but that would mean playing according to the Protocol above will win game $G$ with probability strictly greater than $\cval(G) = 1 - \eps$, which would be a contradiction since the Protocol is a classical strategy. Thus it must be that 
\[
\frac{2 E(\psi) }{\kappa} > \frac{2 \alpha^2 \beta m}{81} > \frac{\alpha^5 n}{5\cdot 90^2 \cdot C} = \frac{(\eps-\gamma)^5 n}{5\cdot 90^2 \cdot \log |\mathcal{A} \times \mathcal{B}|} \;. \qedhere
\]
\end{proof}

\subsection{Proof of Lemmas}\label{sec:proofs_lemmas}

The proofs of Lemmas~\ref{lem:input_dist} and~\ref{lem:corr_samp} are standard in the classical parallel repetition literature, for example,~\cite[Lemmas 4.1 and 6.4]{Hol09}.
%
%
%
%
%

\subsubsection{Proof of Lemma~\ref{lem:bob_ans}}

\emph{Intuition.} All statements we make are within the conditioned event $\theevent$. This Lemma establishes that Bob's $j$'th answer $\bvec_j$ is nearly independent of Alice's question $\xvec_j$ and answer $\avec_j$, conditioned on Bob's question $\yvec_j$ and the dependency-breaking variable $r_{Tj}$. 
We prove this by analyzing Bob's reduced state in the game $G_{1 - \gamma}^n$. If the amount of entanglement used (measured either by the dimension or the entanglement entropy) is too small, then Bob's reduced density matrix cannot have much mutual information with an average $\xvec_j \avec_j$ of Alice. Since Bob's answer $\bvec_j$ is the result of measuring Bob's quantum state, this implies $\bvec_j$ cannot have much mutual information with $\xvec_j \avec_j$ on average. We now formally prove this.

\medskip
\medskip

%
%
%
%
%
%

Given settings $\omega$, $\bvec_S$ of the dependency-breaking variable and Bob's answers in $S$, respectively, define the operator $B_\omega(\bvec_S)$ such that
\[
	B_\omega(\bvec_S)^2 = \Ex_{\yvec | \omega} \sum_{\bvec | \bvec_S} B_\yvec(\bvec)^2.
\]
When we refer to $B_\omega(\bvec_S)$, we refer to the positive square root of the above expression. 

Define the following density matrix on $\density{\Omega \X \A \B_S \Q_B}$:
\[
	\Psi = \Ex_{\omega,\xvec} \sum_{\avec,\bvec_S } \, \dbrac{\omega,\xvec,\avec,\bvec_S}^{\Omega \X \A \B_S} \otimes B_\omega(\bvec_S) \sqrt{\sigma} A_\xvec(\avec)^2 \sqrt{\sigma} B_\omega(\bvec_S)
\]
where the expectation over $\omega, \xvec$ is with respect to the probability measure $\P_{\Omega \Xvec}$, and $\sigma = \Tr_{\Q_A}(\ketbra{\psi}{\psi})$ is the reduced density matrix of $\ket{\psi}$ on Bob's side. The operator $B_\omega(\bvec_S) \sqrt{\sigma} A_\xvec(\avec)^2 \sqrt{\sigma} B_\omega(\bvec_S)$ can be equivalently written as 
\[ \Tr_{\Q_A} \Paren{ (A_\xvec(\avec) \otimes  B_\omega(\bvec_S) )\ketbra{\psi}{\psi} (A_\xvec(\avec) \otimes  B_\omega(\bvec_S))^\dagger}.\] 
The utility of dealing with the operator $B_\omega$ instead of $B_\yvec$ will come from the fact that we do not have to deal with averaging over $\yvec$.

Define $\what{\Psi}$ to be $\Psi$ conditioned on the event $\theevent$. By the fact that $I(\A : \B | \mathsf{C}) \leq H(\B | \mathsf{C})$ we have
\[
	I(\X \A : \Q_B | \Omega, \A_S,\B_S)_{\what{\Psi}} \leq  H(\Q_B | \Omega, \A_S,\B_S)_{\what{\Psi}}.
\]

\begin{claim} $H(\Q_B | \Omega, \A_S,\B_S)_{\what{\Psi}} \leq \frac{H(\Q_B)_\psi}{\P(\theevent)}$.
\end{claim}
We prove this in the Appendix as Claim~\ref{clm:cond_qb}, and assume it for now. We can use the so-called ``randomized chain rule'' for mutual information, where we only consider coordinates $i \in [n]\setminus S$.  We have
\[
	\Ex_{\pi,i} I((\X \A)_{\pi(i)} : \Q_B | \Omega, \A_S,\B_S, (\X\A)_{\pi(< i)})_{\what{\Psi}} \leq  \frac{H(\Q_B)_\psi}{\P(\theevent) m}
\]
where $\pi$ is a random permutation on $[m]$, and $i$ is a uniformly random index in $[m]$.
By Markov's inequality, there exists an $i \leq \beta m$ such that 
\[
\Ex_\pi I((\X \A)_{\pi(i)} : \Q_B | \Omega, \A_S,\B_S, (\X\A)_{\pi(< i)})_{\what{\Psi}} \leq \delta''.
\]
Fix such an $i$.

We can alternatively express the averaging over $\pi$ as first choosing a uniformly random set $T \subseteq [m]$ of size $i - 1 \leq \beta m$, and then choosing a uniformly random $j \in [m] \setminus T$:
\[
\Ex_{T,j\notin T \cup S} I(\X_j \A_j : \Q_B | \Omega, \A_S,\B_S,  (\X\A)_T)_{\what{\Psi}} = \Ex_{T,j \notin T \cup S} \Brac{ \Ex_{\omega, \xvec_T,\avec_\ScupT,\bvec_S | \theevent} I(\X_j \A_j : \Q_B | \omega, \xvec_T, \avec_\ScupT,\bvec_S)_{\what{\Psi}}} \leq \delta''
\]
By further conditioning on $\omega$, fixing Bob's input in coordinate $j$, and bundling the random variables $R_{Tj} = (\Omega_\mj,\Xvec_T,\Avec_\ScupT,\Bvec_S)$ we get that
\begin{equation}
\label{eq:bobstate}
\Ex_{T,j \notin T \cup S} \Brac{ \Ex_{r_{Tj}, \yvec_j | \theevent} I(\X_j \A_j : \Q_B |  r_{Tj}, \yvec_j)_{\what{\Psi}}} \leq 2\delta''.
\end{equation}
Line~\eqref{eq:bobstate} along with Pinsker's inequality yields that 
\begin{equation}
\label{eq:bob_ans}
\Ex_{T,j \notin T \cup S} \Brac{ \Ex_{r_{Tj}, \xvec_j, \yvec_j,\avec_j  | \theevent} \what{\Psi}_{r_{Tj}, \yvec_j}^{\Q_B} \approx_{\sqrt{2\delta''}} \what{\Psi}_{r_{Tj}, \xvec_j,\yvec_j,\avec_j}^{\Q_B} }.
\end{equation}

\medskip
\medskip

To finish the proof of the Lemma we will describe a measurement that when performed on $\what{\Psi}_{r_{Tj}, \yvec_j}^{\Q_B}$ and $\what{\Psi}_{r_{Tj}, \xvec_j,\yvec_j,\avec_j}^{\Q_B} $ respectively produces the probability distributions $\P_{\Bvec_j | r_{Tj},\yvec_j, \theevent}$ and $\P_{\Bvec_j | r_{Tj},\xvec_j,\avec_j,\yvec_j, \theevent}$. Given this, if we apply the measurement on both sides of the approximation of~\eqref{eq:bob_ans}, we get that $\P_{\Bvec_j | r_{Tj},\yvec_j, \theevent}$ and $\P_{\Bvec_j | r_{Tj},\xvec_j,\avec_j,\yvec_j, \theevent}$ are close in statistical distance, concluding the proof.

Notice that since the event $\theevent$ is determined by the variables $\Omega,\Avec_S,\Bvec_S$, for $r_{Tj}$ sampled conditioned on $\theevent$, we have $\what{\Psi}_{r_{Tj}, \xvec_j,\avec_j,  \yvec_j}^{\Q_B} = \Psi_{r_{Tj}, \xvec_j,\avec_j,  \yvec_j}^{\Q_B}$. Furthermore, we have
\[
\Psi_{r_{Tj}, \xvec_j,\avec_j,  \yvec_j}^{\Q_B} = \frac{B_{\omega} (\bvec_S) \sqrt{\sigma} A_{\omega_\mj,\xvec_j}(\avec_{S \cup T \cup \{j\}})^2 \sqrt{\sigma} B_{\omega} (\bvec_S) }{\P(\avec_{S \cup T \cup \{j\}},\bvec_S | \omega_\mj, \xvec_j,\yvec_j)} 
\]
where $\omega$ is defined as $\omega_\mj$ with $\yvec_j$ fixed in the $j$'th coordinate, and we define the operator
\[
A_{\omega_\mj,\xvec_j}(\avec_{S \cup T \cup \{j\}})^2 = \Ex_{\xvec | \omega_\mj,\xvec_j} \sum_{\avec | \avec_{S \cup T \cup \{j\}}} A_\xvec(\avec)^2.
\]
One can verify that the normalization is correct via the following calculation:
	\begin{align*}
	&\Tr \Paren{B_{\omega} (\bvec_S)^2 \sqrt{\sigma} A_{\omega_\mj,\xvec_j}(\avec_{S \cup T \cup \{j\}})^2 \sqrt{\sigma}} \\
	&= \Ex_{\xvec |\omega_\mj, \xvec_j}\,\, \Ex_{\yvec|\omega} \sum_{\avec,\bvec | \avec_{S \cup T \cup \{j\}},\bvec_S} \Tr \Paren{B_\yvec(\bvec)^2 \sqrt{\sigma} A_\xvec(\avec)^2 \sqrt{\sigma}} \\
	&= \Ex_{\xvec, \yvec |\omega_\mj, \xvec_j, \yvec_j}\,\, \sum_{\avec,\bvec | \avec_{S \cup T \cup \{j\}},\bvec_S} \bra{\psi} A_\xvec(\avec)^2 \otimes B_\yvec(\bvec)^2  \ket{\psi} \\	
	&= \Ex_{\xvec, \yvec |\omega_\mj, \xvec_j, \yvec_j} \,\, \sum_{\avec,\bvec | \avec_{S \cup T \cup \{j\}},\bvec_S} \P(\avec,\bvec | \xvec, \yvec) \\
	&= \,\, \P(\avec_{S \cup T \cup \{j\}},\bvec_S | \omega_\mj ,\xvec_j, \yvec_j).
	\end{align*}
%
%
%
%

\medskip
\medskip

\noindent \textbf{The measurement.} For every $\omega_\mj, \yvec_j, \bvec_S$, define the POVM indexed by $\bvec_j$:
\[
	M(\bvec_j) = B_{\omega_\mj,\yvec_j} (\bvec_S)^{-1} B_{\omega_\mj,\yvec_j} (\bvec_{S \cup \{j\}})^2  B_{\omega_\mj,\yvec_j} (\bvec_S)^{-1} 
\]
The operators $B_{\omega_\mj,\yvec_j} (\bvec_S)$ and $B_{\omega_\mj,\yvec_j} (\bvec_{S \cup \{j\}})$ are defined analogously to $A_{\omega_\mj,\xvec_j}(\avec_{S \cup T \cup \{j\}})$.
Thus we get
\begin{align*}
	\Tr \Paren{ M(\bvec_j) \, \Psi_{r_{Tj}, \xvec_j, \yvec_j, \avec_j}^{\Q_B} } &= \frac{\P(\avec_{S \cup T \cup \{j\}},\bvec_{S \cup \{j\}} | \omega_\mj, \xvec_j,\yvec_j)}{\P(\avec_{S \cup T \cup \{j\}},\bvec_S | \omega_\mj, \xvec_j,\yvec_j)} \\
	&= \P(\bvec_j| r_{Tj}, \xvec_j,\avec_j,\yvec_j) \\
	&= \P(\bvec_j| r_{Tj}, \xvec_j,\avec_j,\yvec_j,\theevent)
\end{align*}
where the last equality follows from the fact that the event $\theevent$ is determined by $r_{Tj}$. Similarly, $\Tr \Paren{ M(\bvec_j) \, \Psi_{r_{Tj}, \yvec_j }^{\Q_B} } = \P(\bvec_j | r_{Tj}, \yvec_j, \theevent)$.

\subsubsection{Proof of Lemma~\ref{lem:alice_ans}}
\emph{Intuition.} This statement of this Lemma is very similar to the previous one, except it is simpler in that it argues that for an average coordinate $j$, Alice's answer $\avec_j$ is nearly independent of Bob's question $\yvec_j$, conditioned on $\theevent$ and Alice's question $\xvec_j$ --- notice that we do not consider Bob's answer $\bvec_j$. If we did not condition on the event $\theevent$, this statement would be true exactly because of the no-signaling principle. This Lemma shows that the no-signalling condition for an average coordinate $j$ approximately holds even after conditionin on $\theevent$.

\medskip
\medskip

Fix a set $T$ of size at most $\beta m$. For every $\omega, \xvec_T$, consider the probability distribution $\P_{\Yvec \Avec \Bvec_S | \omega, \xvec_T}$. Notice that
\begin{align}
	\P_{\Yvec \Avec \Bvec_S | \omega, \xvec_T} \preceq (\dim \Avec_\ScupT \Bvec_S) \Paren{\P_{\Yvec  \Avec_{-\ScupT} | \omega, \xvec_T} \cdot \U_{\Avec_\ScupT \Bvec_S}}
\end{align}
where $\U_{\Avec_\ScupT \Bvec_S}$ denotes the uniform distribution over $\Avec_\ScupT \Bvec_S$ and ``$\preceq$'' denotes stochastic dominance. Note that the no-signaling principle implies that
\[
\P_{\Yvec \Avec_{-\ScupT} | \omega, \xvec_T} = \P_{\Yvec | \omega, \xvec_T} \cdot  \P_{\Avec_{-\ScupT} | \omega, \xvec_T} = \P_{\Yvec | \omega} \cdot  \P_{ \Avec_{-\ScupT} | \omega, \xvec_T}.
\]
Therefore
\begin{align}
	D_\infty \Paren{\P_{\Yvec  \Avec \Bvec_S | \omega, \xvec_T} \DivMid \P_{\Yvec | \omega} \cdot \P_{ \Avec_{-\ScupT} |\omega, \xvec_T} \cdot \U_{\Avec_\ScupT \Bvec_S} } \leq \log \dim \Avec_\ScupT \Bvec_S.
\end{align}
We have that by Lemmas~\ref{lem:div_cond} and~\ref{lem:div_chain_rule}, 
\[
\Ex_{\omega,\xvec_T | \theevent} D \Paren{\P_{\Yvec  \Avec  \Bvec_S | \omega, \xvec_T, \theevent} \DivMid \P_{\Yvec  \Avec \Bvec_S |\omega, \xvec_T}} \leq D \Paren{\P_{\Omega \Yvec  \Avec \Bvec_S | \theevent} \DivMid \P_{\Omega \Yvec  \Avec \Bvec_S}} \leq \log \frac{1}{\P(\theevent)}.
\]
Putting everything together, we have
\begin{align}
	&\Ex_{\omega,\xvec_T,\avec_\ScupT,\bvec_S | \theevent} D \Paren{ \P_{\Yvec  \Avec_{-\ScupT} | \omega,\xvec_T, \avec_\ScupT, \bvec_S, \theevent} \DivMid \P_{\Yvec | \omega} \cdot \P_{ \Avec_{-\ScupT} | \omega, \xvec_T, \avec_\ScupT}} \\
&\leq \Ex_{\omega,\xvec_T | \theevent} D \Paren{ \P_{\Yvec  \Avec \Bvec_S | \omega,\xvec_T \theevent} \DivMid \P_{\Yvec | \omega} \cdot \P_{ \Avec_{-\ScupT} | \omega,\xvec_T, \avec_\ScupT} \cdot \U_{\Avec_\ScupT \Bvec_S}} \\
	&\leq \Ex_{\omega,\xvec_T | \theevent} D \Paren{\P_{\Yvec  \Avec \Bvec_S | \omega,\xvec_T, \theevent} \DivMid \P_{\Yvec  \Avec \Bvec_S | \omega,\xvec_T}} 	\\
	&\qquad \qquad \qquad + D_\infty \Paren{\P_{\Yvec  \Avec \Bvec_S | \omega,\xvec_T} \DivMid \P_{\Yvec | \omega} \cdot \P_{ \Avec_{-\ScupT} | \omega,\xvec_T} \cdot \U_{\Avec_\ScupT \Bvec_S} } \\
	&\leq \log \frac{1}{\P(\theevent)} + \log \dim (\Avec_{S \cup T} \Bvec_S).
\end{align}

Notice that $\P_{\Yvec | \omega}$ is a product distribution across the coordinates of $\Yvec$. Therefore Raz's Lemma applies, and we get that
\begin{align}
	&\frac{1}{m - |T|} \Ex_{\omega,\xvec_T,\avec_\ScupT,\bvec_S | \theevent} \sum_{j \notin T \cup S} I(\Yvec_j : \Avec_{-\ScupT} | \omega,\xvec_T,\avec_\ScupT,\bvec_S) \\
	&\leq \frac{1}{m - \beta m} \Paren{\log \frac{1}{\P(\theevent)} + \log \dim (\Avec_{S \cup T} \Bvec_S)} = \delta'.
\end{align}
By conditioning on $\omega$ fixing Alice's input in coordinate $j$, we get that
\begin{equation}
\label{eq:alicestate}
	\Ex_{j \notin T \cup S} \,\, \Ex_{\xvec_j,\omega_\mj,\xvec_T,\avec_\ScupT,\bvec_S | \theevent}  I(\Yvec_j : \Avec_j | \xvec_j, \omega_\mj,\xvec_T,\avec_\ScupT,\bvec_S) \leq 2\delta'.
\end{equation}
Bundling the random variables $\R_{Tj} = (\Omega_\mj,\Xvec_T,\Avec_\ScupT,\Bvec_S)$ and by Pinsker's inequality, we have
\begin{equation}
\label{eq:alice2}
		\Ex_{j \notin T \cup S} \,\,\, \Ex_{r_{Tj},\xvec_j,\yvec_j | \theevent} \Brac{\P_{\Avec_j | r_{Tj}, \xvec_j, \yvec_j, \theevent} \approx_{\sqrt{2\delta'}} \P_{\Avec_j |  r_{Tj}, \xvec_j, \theevent} }.
\end{equation}
This concludes the proof.

\subsection{Proof of Theorem~\ref{cor:lwb_ent_single_game}}

We now explain how Theorem~\ref{cor:lwb_ent_single_game} easily follows from our main theorem, Theorem~\ref{thm:main}. We first restate the theorem.

\begin{customcorollary}{\ref{cor:lwb_ent_single_game}}
	Let $G$ be a two-player game with a classical-quantum gap: i.e., $\Delta := \qval(G) - \cval(G) > 0$. Let $0 \leq \nu < \Delta$ be a noise parameter. Then, for any state $\sigma$ that can be used to win~$G$ with probability at least $\qval(G) - \nu$, its entanglement cost satisfies ]$E_C(\sigma)\geq c_2/4$, where~$c_2$ is the constant from Theorem~\ref{thm:main} as given in Equation~\eqref{eq:final_thm_constants}.
\end{customcorollary}

\begin{proof}
	Consider a quantum strategy for $G$ that uses a state $\sigma$ that succeeds with probability $\qval(G) - \nu$.
	Playing $n$ instances of the considered strategy in parallel, using the state~$\sigma^{\otimes n}$, will succeed in the threshold game $G^n_{\qval(n) - \nu}$ with probability $1/2$. Hence, by Theorem~\ref{thm:main}, 
	\[
		E_F(\sigma^{\otimes n}) \geq c_2 n/4 \;.
	\]
	According to~\cite{hayden2001asymptotic}, the entanglement cost is equal to the \emph{regularised} entanglement of formation:
	\[
		E_C(\rho) = E^{\infty}_F(\rho)=\lim_{n\rightarrow\infty} (E_F(\rho^{\otimes n})/n) \;.
	\]
	Combining the above two observations together we get $E_C(\sigma)\geq c_2/4$.
\end{proof} 

\appendix

\section{Omitted proofs}

\begin{proof}[Proof of Proposition~\ref{prop:subset2}]
	Let $\alpha = \eps - \gamma$, and let $\tau,\delta,t$ be parameters that we will choose later, subject to $\gamma < \tau < \delta$. 

	We first show that $\Ex_S \Pr(\neg W_j | \theevent) \leq \eps - \alpha/4$, where $S$ is a (multi)set of $t$ independently chosen indices in $[n]$. 
	First we write, for a fixed $S$,
	\begin{align*}
		\Pr ( \neg W_j | \theevent) &= \Pr(\neg W_j | \theevent, W^{> 1 - \delta}) \Pr(W^{> 1 - \delta} | \theevent) + \\ &\qquad \qquad \Pr(\neg W_j | \theevent, \neg W^{> 1 - \delta}) \Pr(\neg W^{> 1 - \delta} | \theevent) \\
		&\leq \Pr(\neg W_j | \theevent \wedge W^{> 1 - \delta}) + \Pr(\neg W^{> 1 - \delta} | \theevent)
	\end{align*}
	Observe that $\Pr(\neg W_j | \theevent \wedge W_{> 1 - \delta})$ is the probability that, conditioned on winning all rounds in $S$, the randomly selected coordinate $j \in [n] - S$ happens to be one of the (at most) $\delta n$ lost rounds. This is at most $\delta n/(n - t)$. 
	
	Next we bound $\Ex_S \Pr(\neg W^{> 1 - \delta} | \theevent)$. Call a subset $S$ \emph{good} if $\Pr(W_S^{\geq 1 - \tau}) \geq \Pr(W^{\geq 1 - \gamma})/2$, and \emph{bad} otherwise. We now bound the probability a uniform random $S$ is good. Let $R$ denote the random subset of $[n]$ that indicates which rounds were won. Then  for every $R$ such that least $|R| \geq (1 - \gamma)n$, the probability conditioned on $R$ of picking $S$ such that $W_S^{\geq 1 - \tau} = 1$ is at least $1 - \exp(- (\gamma - \tau)^2 t/3)$, by a standard Chernoff bound. Therefore
	$$
	\Ex_{R | W^{\geq 1 - \gamma}} \Ex_S \,\, \Pr(W_S^{\geq 1 -\tau} | R) =\Ex_{R | W^{\geq 1 - \gamma}} \Ex_S\,\, \Ind \{ W_S^{\geq 1 -\tau} | R \} \geq 1 - \exp(- (\gamma - \tau)^2  t/3)
	$$
	where the expectation over $R | W^{\geq 1 - \gamma}$ denotes picking $R$ conditioned on the fact that $|R| \geq (1 - \gamma)n$. By Markov's inequality, this implies that at least a $1 - \sqrt{\exp(- (\gamma - \tau)^2  t/3)}$ fraction of $S$'s are such that
	$$
		\Pr(W_S^{\geq 1 - \tau} | W^{\geq 1 - \gamma}) = \Ex_{R | W^{\geq 1 - \gamma}} \Pr(W_S^{\geq 1 -\tau} | R) \geq 1 - \sqrt{\exp(- (\gamma - \tau)^2 t/3)}.
	$$
	Notice that $\Pr(W_S^{\geq 1 - \tau}) \geq \Pr(W_S^{\geq 1 - \tau} | W^{\geq 1 - \gamma}) \cdot \Pr(W^{\geq 1 - \gamma})$, so therefore $\exp(- (\gamma - \tau)^2 t/6)$ bounds the probability that $S$ is bad. 
	
	Now observe that 
	\begin{align*}
		\Ex_S \Pr(\neg W^{\geq 1 - \delta} | \theevent) &\leq \Pr(\text{$S$ bad}) + \sum_{\text{$S$ good}} \Pr(S) \cdot \Pr(\neg W^{\geq 1 - \delta} | \theevent) \\
		&\leq \Pr(\text{$S$ bad}) + \sum_{\text{$S$ good}} \Pr(S) \cdot \frac{\Pr(\theevent | \neg W^{\geq 1 - \delta})}{\Pr(\theevent)} \\
								   &\leq \Pr(\text{$S$ bad}) + \frac{2}{\Pr(W^{\geq 1 - \gamma})} \sum_{\text{$S$ good}} \Pr(S) \cdot \Pr(\theevent | \neg W^{> 1 - \delta})
	\end{align*}
	where we used the definition of $S$ being good in the third inequality. To bound the second term of the last line, we compute
	\begin{align*}
	\sum_{\text{$S$ good}} \Pr(S) \cdot \Pr(\theevent | \neg W^{> 1 - \delta})
&\leq \Ex_S \Pr(\theevent | \neg W^{> 1 - \delta}) \\
&= \Ex_{R | \neg W^{> 1 - \delta}} \Ex_S \,\, \Pr(\theevent | R) \\
&= \Ex_{R | \neg W^{> 1 - \delta}} \Ex_S \,\, \Ind \{ \theevent | R \}
	\end{align*}
where the expectation over $R |\neg W^{\leq 1 - \delta}$ denotes picking $R$ conditioned on $|R| \leq (1 - \delta)n$. By a Chernoff bound, for every such $R$, the probability of picking $S$ such that $W_S^{\geq 1 - \tau} = 1$ is at most $\exp(- (\delta - \tau)^2 t/3)$. 

Choose $\tau = \eps - \frac{3}{4}\alpha$ and let $\Delta = \alpha/4$. Let $\kappa$ denote $\Pr(W^{\geq 1 - \gamma})$. Set
\[
	t = \frac{6}{\Delta^2} \Paren{ \ln \frac{2}{\kappa} + \ln \frac{8}{\alpha}}.
\]
Finally, let $\delta = \eps - \frac{1}{4} \alpha - \frac{t}{n}$. Since we are assuming that $\kappa \geq \frac{16}{\alpha} \exp \Paren{ - \frac{\alpha^3}{384} n}$, this implies that $t/n \leq \alpha/4$. 

Putting everything together, we get that
\begin{align*}
\Ex_S \Pr ( \neg W_j | \theevent) &\leq \frac{\delta n}{n - t} + e^{- (\gamma - \tau)^2 t/6} + \frac{2}{\kappa}  e^{- (\delta - \tau)^2 t/3} \\ 
	&\leq \frac{\delta n}{n - t} + e^{- \Delta^2 t/6} + \frac{2}{\kappa}  e^{- \Delta^2 t/3} \\
	&\leq \frac{\delta n}{n - t} + \frac{\alpha}{4} \\
	&\leq \alpha/2
\end{align*}
by our choice of parameters.

Therefore by the probabilistic method there is a $S$ that satisfies the conclusions of the Proposition statement.
\end{proof}

\begin{claim}
\label{clm:cond_qb}
\[
H(\Q_B | \Omega, \A_S, \B_S)_{\what{\Psi}} \leq \frac{H(\Q_B)_\psi}{\P(\theevent)}.
\]
\end{claim}

\begin{proof}
	Let $\lambda = \P(\theevent)$. Recall that the state $\Psi$ is a density matrix on registers $\density{\Omega \X \A \B_S \Q_B}$ that is classical on $\Omega \X \A \B_S$, and that $\what{\Psi}$ is $\Psi$ conditioned on the event $\theevent$, which is solely a function of the classical registers $\Omega \A_S \B_S$. Therefore $\Psi^{\Omega \X \A \B_S \Q_B} = \lambda \cdot \Psi^{\Omega \X \A \B_S \Q_B}_{\theevent} + (1 - \lambda) \cdot \Psi^{\Omega \X \A \B_S \Q_B}_{\neg \theevent}$ and hence
	\begin{align*}
	\lambda \cdot H(\Q_B | \Omega, \A_S, \B_S)_{\what{\Psi}} &\leq \lambda \cdot H(\Q_B | \Omega, \A_S,\B_S, \theevent)_{\Psi} +  (1 - \lambda) \cdot H(\Q_B | \Omega, \A_S,\B_S, \neg \theevent)_\Psi \\
		&= H(\Q_B | \Omega, \A_S, \B_S, \mathsf{F})_\Psi
	\end{align*}
	where $\mathsf{F}$ is an additional qubit register that stores the result of the binary measurement that checks whether $\theevent$ happened. Since $\Omega, \A_S, \B_S$ are classical, introducing this extra qubit register does not change the density matrix of $\Psi$ on $\Omega \X \A \B_S \Q_B$. 
	
	Since conditioning can only decrease entropy, we have that the above is at most $H(\Q_B | \Omega, \B_S)_\Psi \leq H(\Q_B \B_S | \Omega)_\Psi$. Notice that for every $\omega$, the density matrix $\Psi^{\Q_B \B_S}_\omega$ is the following:
	\begin{align*}
\Psi^{\Q_B \B_S}_\omega &=	\Ex_{\xvec} \sum_{\avec,\bvec_S}  \Paren{B_\omega(\bvec_S) \sqrt{\sigma} A_\xvec(\avec)}^2 \otimes \dbrac{\bvec_S}^{\B_S} \\
&= \sum_{\bvec_S} \Paren{B_\omega(\bvec_S) \sqrt{\sigma}}^2 \otimes \dbrac{\bvec_S}^{\B_S}.
	\end{align*}
But now notice that $\Psi^{\Q_B \B_S}_\omega$ is unitarily equivalent to to 
\[
	\Phi^{\Q_B \B_S} = \sigma^{\Q_B} \otimes \dbrac{0}^{\B_S}.
\]
Since von Neumann entropy is invariant under unitary transformations, we have
$$
	H(\Q_B \B_S | \Omega)_\Psi = H(\Q_B \B_S | \Omega)_\Phi = H(\Q_B)_\rho.
$$
Since $\rho = \psi^{\Q_B}$, this completes the proof.
\end{proof}

\bibliography{robust_dim_testing}

\end{document}

%% file: qmacros.tex
\usepackage{stmaryrd}

\newcommand{\ket}[1]{|#1\rangle}
\newcommand{\bra}[1]{\langle#1|}

\newcommand{\ketbra}[2]{|#1\rangle\! \langle #2|}
\newcommand{\Tr}{\mbox{\rm Tr}}
\newcommand{\Id}{\mathrm{Id}}

\DeclareMathOperator*{\Ex}{\mathbb{E}}

\newcommand{\Ind}{\mathbf{1}}
\newcommand{\C}{\ensuremath{\mathbb{C}}}

\newcommand{\R}{\ensuremath{\mathbb{R}}}

\renewcommand{\P}{\mathsf{P}}
\newcommand{\Q}{\mathsf{Q}}

\newcommand{\setft}[1]{\mathrm{#1}}

\newcommand{\density}[1]{\setft{Dens}\!\left(#1\right)}

\DeclareMathOperator{\poly}{poly}

\newcommand{\eps}{\varepsilon}

\newcommand{\paren}[1]{(#1)}
\newcommand{\Paren}[1]{\left(#1\right)}

\newcommand{\brac}[1]{[#1]}
\newcommand{\Brac}[1]{\left[#1\right]}

\newcommand{\dbrac}[1]{\llbracket #1 \rrbracket}

\newcommand{\set}[1]{\{#1\}}

\newcommand{\norm}[1]{\lVert#1\rVert}
\newcommand{\Norm}[1]{\left\lVert#1\right\rVert}







%% file: robust_dim_testing.bbl
\newcommand{\etalchar}[1]{$^{#1}$}
\begin{thebibliography}{SMSC{\etalchar{+}}15}

\bibitem[AFB17]{arnon2017noise}
Rotem Arnon-Friedman and Jean-Daniel Bancal.
\newblock Device-independent certification of one-shot distillable
  entanglement.
\newblock {\em arXiv}, 2017.

\bibitem[AFRV16]{arnon2016simple}
Rotem Arnon-Friedman, Renato Renner, and Thomas Vidick.
\newblock Simple and tight device-independent security proofs.
\newblock {\em arXiv preprint arXiv:1607.01797}, 2016.

\bibitem[AVDM01]{audenaert2001variational}
Koenraad Audenaert, Frank Verstraete, and Bart De~Moor.
\newblock Variational characterizations of separability and entanglement of
  formation.
\newblock {\em Physical Review A}, 64(5):052304, 2001.

\bibitem[BBPS96]{bennett1996concentrating}
Charles~H Bennett, Herbert~J Bernstein, Sandu Popescu, and Benjamin Schumacher.
\newblock Concentrating partial entanglement by local operations.
\newblock {\em Physical Review A}, 53(4):2046, 1996.

\bibitem[BDSW96]{bennett1996mixed}
Charles~H Bennett, David~P DiVincenzo, John~A Smolin, and William~K Wootters.
\newblock Mixed-state entanglement and quantum error correction.
\newblock {\em Physical Review A}, 54(5):3824, 1996.

\bibitem[Bel64]{bell1964}
John~S Bell.
\newblock On the {Einstein-Podolsky-Rosen} paradox.
\newblock {\em Physics}, 1(3), 1964.

\bibitem[BH10]{brandao2010hastings}
Fernando~GSL Brandao and Micha{\l} Horodecki.
\newblock On hastings' counterexamples to the minimum output entropy additivity
  conjecture.
\newblock {\em Open Systems \& Information Dynamics}, 17(01):31--52, 2010.

\bibitem[BHK05]{barrett2005no}
Jonathan Barrett, Lucien Hardy, and Adrian Kent.
\newblock No signaling and quantum key distribution.
\newblock {\em Physical Review Letters}, 95(1):010503, 2005.

\bibitem[BIS{\etalchar{+}}16]{boixo2016characterizing}
Sergio Boixo, Sergei~V Isakov, Vadim~N Smelyanskiy, Ryan Babbush, Nan Ding,
  Zhang Jiang, John~M Martinis, and Hartmut Neven.
\newblock Characterizing quantum supremacy in near-term devices.
\newblock {\em arXiv preprint arXiv:1608.00263}, 2016.

\bibitem[BPA{\etalchar{+}}08]{brunner2008testing}
Nicolas Brunner, Stefano Pironio, Antonio Acin, Nicolas Gisin, Andr{\'e}~Allan
  M{\'e}thot, and Valerio Scarani.
\newblock Testing the dimension of hilbert spaces.
\newblock {\em Physical review letters}, 100(21):210503, 2008.

\bibitem[BVY16]{BVY15fort}
Mohammad Bavarian, Thomas Vidick, and Henry Yuen.
\newblock Parallel repetition via fortification: analytic view and the quantum
  case.
\newblock {\em arXiv preprint arXiv:1603.05349}, 2016.

\bibitem[BVY17]{BVY17}
Mohammad Bavarian, Thomas Vidick, and Henry Yuen.
\newblock Hardness amplification for entangled games via anchoring.
\newblock In {\em Proceedings of the 49th Annual {ACM} {SIGACT} Symposium on
  Theory of Computing, {STOC} 2017, Montreal, QC, Canada, June 19-23, 2017},
  pages 303--316, 2017.

\bibitem[BYJKS02]{bar2002information}
Ziv Bar-Yossef, Thathachar~S Jayram, Ravi Kumar, and D~Sivakumar.
\newblock An information statistics approach to data stream and communication
  complexity.
\newblock In {\em Foundations of Computer Science, 2002. Proceedings. The 43rd
  Annual IEEE Symposium on}, pages 209--218. IEEE, 2002.

\bibitem[CBRS16]{cai2016new}
Yu~Cai, Jean-Daniel Bancal, Jacquiline Romero, and Valerio Scarani.
\newblock A new device-independent dimension witness and its experimental
  implementation.
\newblock {\em Journal of Physics A: Mathematical and Theoretical},
  49(30):305301, 2016.

\bibitem[CGJV17]{coladangelo2017verifier}
Andrea Coladangelo, Alex Grilo, Stacey Jeffery, and Thomas Vidick.
\newblock Verifier-on-a-leash: new schemes for verifiable delegated quantum
  computation, with quasilinear resources.
\newblock {\em arXiv preprint arXiv:1708.07359}, 2017.

\bibitem[CN16]{coudron2016parallel}
Matthew Coudron and Anand Natarajan.
\newblock The parallel-repeated magic square game is rigid.
\newblock {\em arXiv preprint arXiv:1609.06306}, 2016.

\bibitem[Col17]{coladangelo2017parallel}
Andrea Coladangelo.
\newblock Parallel self-testing of (tilted) epr pairs via copies of (tilted)
  chsh and the magic square game.
\newblock {\em Quantum Information and Computation}, 17(9-10):831--865, 2017.

\bibitem[Coo17]{coopmans2017robust}
Tim Coopmans.
\newblock Robust self-testing of (almost) all pure two-qubit states.
\newblock Master's thesis, Universiteit van Amsterdam, 2017.

\bibitem[CRSV16]{chao2016test}
Rui Chao, Ben~W Reichardt, Chris Sutherland, and Thomas Vidick.
\newblock Test for a large amount of entanglement, using few measurements.
\newblock {\em arXiv preprint arXiv:1610.00771}, 2016.

\bibitem[CS17]{coladangelo2017robust}
Andrea Coladangelo and Jalex Stark.
\newblock Robust self-testing for linear constraint system games.
\newblock {\em arXiv preprint arXiv:1709.09267}, 2017.

\bibitem[CWY15]{chung2015parallel}
Kai-Min Chung, Xiaodi Wu, and Henry Yuen.
\newblock Parallel repetition for entangled k-player games via fast quantum
  search.
\newblock In {\em the 30th Conference on Computational Complexity (CCC)}, pages
  512--536, 2015.

\bibitem[FBB17]{friis2017geometry}
Nicolai Friis, Sridhar Bulusu, and Reinhold~A Bertlmann.
\newblock Geometry of two-qubit states with negative conditional entropy.
\newblock {\em Journal of Physics A: Mathematical and Theoretical},
  50(12):125301, 2017.

\bibitem[GKW15]{gheorghiu2015robustness}
Alexandru Gheorghiu, Elham Kashefi, and Petros Wallden.
\newblock Robustness and device independence of verifiable blind quantum
  computing.
\newblock {\em New Journal of Physics}, 17(8):083040, 2015.

\bibitem[GVW{\etalchar{+}}15]{giustina2015significant}
Marissa Giustina, Marijn~AM Versteegh, S{\"o}ren Wengerowsky, Johannes
  Handsteiner, Armin Hochrainer, Kevin Phelan, Fabian Steinlechner, Johannes
  Kofler, Jan-{\AA}ke Larsson, Carlos Abell{\'a}n, et~al.
\newblock Significant-loophole-free test of bell’s theorem with entangled
  photons.
\newblock {\em Physical review letters}, 115(25):250401, 2015.

\bibitem[HBD{\etalchar{+}}15]{hensen2015loophole}
Bas Hensen, H~Bernien, AE~Dr{\'e}au, A~Reiserer, N~Kalb, MS~Blok, J~Ruitenberg,
  RFL Vermeulen, RN~Schouten, C~Abell{\'a}n, et~al.
\newblock Loophole-free {B}ell inequality violation using electron spins
  separated by 1.3 kilometres.
\newblock {\em Nature}, 526(7575):682--686, 2015.

\bibitem[HHHH09]{horodecki2009quantum}
Ryszard Horodecki, Pawe{\l} Horodecki, Micha{\l} Horodecki, and Karol
  Horodecki.
\newblock Quantum entanglement.
\newblock {\em Reviews of modern physics}, 81(2):865, 2009.

\bibitem[HHT01]{hayden2001asymptotic}
Patrick~M Hayden, Michal Horodecki, and Barbara~M Terhal.
\newblock The asymptotic entanglement cost of preparing a quantum state.
\newblock {\em Journal of Physics A: Mathematical and General}, 34(35):6891,
  2001.

\bibitem[Hol09]{Hol09}
Thomas Holenstein.
\newblock Parallel repetition: Simplification and the no-signaling case.
\newblock {\em Theory of Computing}, 5(8):141--172, 2009.

\bibitem[HPDF15]{hajduvsek2015device}
Michal Hajdu{\v{s}}ek, Carlos~A P{\'e}rez-Delgado, and Joseph~F Fitzsimons.
\newblock Device-independent verifiable blind quantum computation.
\newblock {\em arXiv preprint arXiv:1502.02563}, 2015.

\bibitem[IBM]{IBM}
{IBM} quantum experience.
\newblock \url{https://www.research.ibm.com/ibm-q/}.

\bibitem[JPPG{\etalchar{+}}10]{junge2010unbounded}
Marius Junge, Carlos Palazuelos, David P{\'e}rez-Garc{\'\i}a, Ignacio
  Villanueva, and Michael~M Wolf.
\newblock Unbounded violations of bipartite bell inequalities via operator
  space theory.
\newblock {\em Communications in Mathematical Physics}, 300(3):715--739, 2010.

\bibitem[JPY14]{JainPY14}
Rahul Jain, Attila Pereszl{\'{e}}nyi, and Penghui Yao.
\newblock A parallel repetition theorem for entangled two-player one-round
  games under product distributions.
\newblock In {\em Proceedings of Conference on Computational Complexity (CCC)},
  pages 209--216, 2014.

\bibitem[LYL{\etalchar{+}}17]{liu2017high}
Yang Liu, Xiao Yuan, Ming-Han Li, Weijun Zhang, Qi~Zhao, Jiaqiang Zhong, Yuan
  Cao, Yu-Huai Li, Luo-Kan Chen, Hao Li, et~al.
\newblock High speed self-testing quantum random number generation without
  detection loophole.
\newblock In {\em Frontiers in Optics}, pages FTh2E--1. Optical Society of
  America, 2017.

\bibitem[McK16]{mckague2016self}
Matthew McKague.
\newblock Self-testing in parallel.
\newblock {\em New Journal of Physics}, 18(4):045013, 2016.

\bibitem[MS14]{miller2014robust}
Carl~A Miller and Yaoyun Shi.
\newblock Robust protocols for securely expanding randomness and distributing
  keys using untrusted quantum devices.
\newblock In {\em Proceedings of the 46th Annual ACM Symposium on Theory of
  Computing}, pages 417--426. ACM, 2014.

\bibitem[NC10]{nielsen2010quantum}
Michael~A Nielsen and Isaac~L Chuang.
\newblock {\em Quantum computation and quantum information}.
\newblock Cambridge university press, 2010.

\bibitem[NV17]{natarajan2017quantum}
Anand Natarajan and Thomas Vidick.
\newblock A quantum linearity test for robustly verifying entanglement.
\newblock In {\em Proceedings of the 49th Annual ACM SIGACT Symposium on Theory
  of Computing}, pages 1003--1015. ACM, 2017.

\bibitem[PAB{\etalchar{+}}09]{pironio2009device}
Stefano Pironio, Antonio Ac{\'\i}n, Nicolas Brunner, Nicolas Gisin, Serge
  Massar, and Valerio Scarani.
\newblock Device-independent quantum key distribution secure against collective
  attacks.
\newblock {\em New Journal of Physics}, 11(4):045021, 2009.

\bibitem[PR97]{popescu1997thermodynamics}
Sandu Popescu and Daniel Rohrlich.
\newblock Thermodynamics and the measure of entanglement.
\newblock {\em Physical Review A}, 56(5):R3319, 1997.

\bibitem[PV05]{plenio2005introduction}
Martin~B Plenio and Shashank Virmani.
\newblock An introduction to entanglement measures.
\newblock {\em arXiv preprint quant-ph/0504163}, 2005.

\bibitem[PV09]{pal2009quantum}
K{\'a}roly~F P{\'a}l and Tam{\'a}s V{\'e}rtesi.
\newblock Quantum bounds on bell inequalities.
\newblock {\em Physical Review A}, 79(2):022120, 2009.

\bibitem[Rao11]{rao2011parallel}
Anup Rao.
\newblock Parallel repetition in projection games and a concentration bound.
\newblock {\em SIAM Journal on Computing}, 40(6):1871--1891, 2011.

\bibitem[Raz98]{raz1998parallel}
Ran Raz.
\newblock A parallel repetition theorem.
\newblock {\em SIAM Journal on Computing}, 27(3):763--803, 1998.

\bibitem[RUV13]{reichardt2013classical}
Ben~W Reichardt, Falk Unger, and Umesh Vazirani.
\newblock Classical command of quantum systems.
\newblock {\em Nature}, 496(7446):456--460, 2013.

\bibitem[SMSC{\etalchar{+}}15]{shalm2015strong}
Lynden~K Shalm, Evan Meyer-Scott, Bradley~G Christensen, Peter Bierhorst,
  Michael~A Wayne, Martin~J Stevens, Thomas Gerrits, Scott Glancy, Deny~R
  Hamel, Michael~S Allman, et~al.
\newblock Strong loophole-free test of local realism.
\newblock {\em Physical review letters}, 115(25):250402, 2015.

\bibitem[TV00]{terhal2000entanglement}
Barbara~M Terhal and Karl Gerd~H Vollbrecht.
\newblock Entanglement of formation for isotropic states.
\newblock {\em Physical Review Letters}, 85(12):2625, 2000.

\bibitem[VB14]{vertesi2014disproving}
Tam{\'a}s V{\'e}rtesi and Nicolas Brunner.
\newblock Disproving the peres conjecture: Bell nonlocality from bipartite
  bound entanglement.
\newblock {\em arXiv preprint arXiv:1405.4502}, 2014.

\bibitem[VV14]{vazirani2014fully}
Umesh Vazirani and Thomas Vidick.
\newblock Fully device-independent quantum key distribution.
\newblock {\em Physical review letters}, 113(14):140501, 2014.

\bibitem[VW01]{vollbrecht2001entanglement}
Karl Gerd~H Vollbrecht and Reinhard~F Werner.
\newblock Entanglement measures under symmetry.
\newblock {\em Physical Review A}, 64(6):062307, 2001.

\bibitem[VW02]{verstraete2002entanglement}
Frank Verstraete and Michael~M Wolf.
\newblock Entanglement versus bell violations and their behavior under local
  filtering operations.
\newblock {\em Physical review letters}, 89(17):170401, 2002.

\bibitem[Wil13]{wilde2013quantum}
Mark~M Wilde.
\newblock {\em Quantum information theory}.
\newblock Cambridge University Press, 2013.

\bibitem[Woo98]{wootters1998entanglement}
William~K Wootters.
\newblock Entanglement of formation of an arbitrary state of two qubits.
\newblock {\em Physical Review Letters}, 80(10):2245, 1998.

\bibitem[Woo01]{wootters2001entanglement}
William~K Wootters.
\newblock Entanglement of formation and concurrence.
\newblock {\em Quantum Information \& Computation}, 1(1):27--44, 2001.

\bibitem[Woo02]{wotters2002ent}
William~K. Wootters.
\newblock Entanglement of formation.
\newblock In Prem Kumar, Giacomo~M D'Ariano, and Osamu Hirota, editors, {\em
  Quantum Communication, Computing and Measurement 2}, pages 69--74. Springer,
  2002.

\bibitem[Yue16]{yuen2016parallel}
Henry Yuen.
\newblock A parallel repetition theorem for all entangled games.
\newblock In {\em 43rd International Colloquium on Automata, Languages, and
  Programming, {ICALP} 2016, July 11-15, 2016, Rome, Italy}, pages 77:1--77:13,
  2016.

\end{thebibliography}
